\documentclass[12pt]{article}
\usepackage[margin=1in]{geometry}

\usepackage[english]{babel}
\usepackage[utf8x]{inputenc}
\usepackage{amsmath}
\usepackage{amssymb}
\usepackage{graphicx}
\usepackage[colorinlistoftodos]{todonotes}
\usepackage{float}
\usepackage{hyperref}
\hypersetup{
    colorlinks=true,
    linkcolor=blue,
    filecolor=magenta,  
    citecolor=blue,
    urlcolor=blue}
\usepackage{natbib}
\usepackage{bm}
\usepackage{multirow}
\usepackage[toc,page]{appendix}
\usepackage{xcolor}

\usepackage[inline]{enumitem}

\usepackage{caption}
\usepackage{subcaption}
\usepackage{rotating}
\usepackage{amsthm}
\theoremstyle{plain}
\newtheorem{theorem}{Theorem}[section]
\newtheorem{corollary}{Corollary}[theorem]
\newtheorem{remark}{Remark}

\title{\textbf{Controlling the flexibility of non-Gaussian processes through shrinkage priors} }
\date{\vspace{-5ex}}

\author{Rafael Cabral, David Bolin and H\aa vard Rue \\ \small{\it{King Abdullah University of Science and Technology, Thuwal, Saudi Arabia}}}

\begin{document}
\maketitle

\begin{abstract}

The normal inverse Gaussian (NIG) and generalized asymmetric Laplace (GAL) distributions can be seen as skewed and semi-heavy-tailed extensions of the Gaussian distribution. Models driven by these more flexible noise distributions are then regarded as flexible extensions of simpler Gaussian models. 
Inferential procedures tend to overestimate the degree of non-Gaussianity in the data and therefore we propose controlling the flexibility of these non-Gaussian models by adding sensible priors in the inferential framework that contract the model towards Gaussianity. In our venture to derive sensible priors, we also propose a new intuitive parameterization of the non-Gaussian models and  discuss how to implement them efficiently in $Stan$. The methods are derived for a generic class of non-Gaussian models that include spatial Matérn fields, autoregressive models for time series, and simultaneous autoregressive models for aerial data. The results are illustrated with a simulation study and geostatistics application, where priors that penalize model complexity were shown to lead to more robust estimation and give preference to the Gaussian model, while at the same time allowing for non-Gaussianity if there is sufficient evidence in the data. 
\end{abstract}


\vspace{2ex}
\textbf{MSC2020 subject classifications}: Primary 62F15, 62M20; secondary 62M40.

\section{Introduction}


Gaussian processes are the most common class of models to describe spatial and temporal dependence in Bayesian hierarchical models. Due to their well-established theory, flexibility, and practicality, Gaussian processes are also a fundamental building block in spatial and temporal statistics. However, additional flexibility is needed in several applications, and Gaussian processes are replaced by more flexible non-Gaussian Lévy processes \citep{ken1999levy}. In this paper, we study processes driven by generalized hyperbolic (GH) noise, more specifically the normal inverse Gaussian (NIG) and generalized asymmetric Laplace (GAL) subfamilies, which contain the Gaussian distribution as a particular case and are semi-heavy-tailed. 


\subsection{Literature review} \label{sect:litreview}

Autoregressive processes with GH innovating terms are discussed in \cite{GHAR} and can model extreme market movements not captured by the Gaussian model \citep{NIGAR}.  \cite{BIBBY2003211} presented more sophisticated models for financial time series, namely stochastic processes whose marginal distributions or the distributions of increments (or both) are generalized hyperbolic, including classical diffusion processes and stochastic volatility models \citep{barnig,nongaussianbarnsdorf,bareco}. The empirical distributions of the log-returns of financial time series are too heavy for satisfactory fitting by normal densities, and the GH distribution can better capture the stylized features of financial data. 

In the field of spatial statistics, \cite{bolin2014spatial} provided a class of non-Gaussian random fields with Matérn covariance function, constructed as solutions to stochastic partial differential equations (SPDEs) driven by Generalized Hyperbolic noise. The previous models were applied in the context of geostatistics \citep{wallin2015geostatistical}, joint modeling of multivariate random fields \citep{bolin2020multivariate}, 
and to model continuously repeated measurement data collected longitudinally \citep{asar2020linear}, where the non-Gaussian models led to improved predictive power in datasets where there were sharp spikes or jumps in the observed data that the Gaussian models oversmoothed. 

The aforementioned papers performed parameter estimation of non-Gaussian models via likelihood maximization techniques or by the method of moments. Bayesian estimation was carried out for the same geostatistical framework of \cite{wallin2015geostatistical} by \cite{walder2020bayesian}. A Bayesian framework has also been considered for stochastic volatility models \citep{sv1,sv2}, for a vector of autoregressive processes \citep{karlsson2021vector}, and GARCH models \citep{GHSSTGH}, and all previously cited papers utilized the generalized hyperbolic skew Student’s $t$ (GHSST) distribution for the driving noise. In the same papers, a gamma prior was chosen for the leptokurtosis parameter and a normal prior for the skewness parameter of the GHSST distribution. There was no principled motivation for the use of these priors other than mathematical or computational convenience, and the gamma prior places 0 mass at the base Gaussian model, thus forcing overfitting of the data \citep{simpson2017penalising}.




\subsection{Motivation and contributions}


Here we present the layout of the paper and highlight what are its main contributions to the literature on non-Gaussian processes from a Bayesian perspective. We find in the current literature challenges when it comes to model interpretability, sensible prior selection, and ease of implementation which we address next.


\emph{Parameterization and interpretability}: Parameterizations of the NIG and GAL distributions involve 4 parameters that are used to regulate the mean, variance, skewness, and kurtosis of the distribution. However, the parameterizations are not themselves property-based leading to a difficult interpretation of the parameters. This also presents a difficulty when constructing priors for the parameters, since, for instance, the mean and variance will depend not only on the location and scale parameter but also on the other parameters. In section \ref{sect:distributions} we provide a property-based parameterization of the NIG and GAL distributions that preserves the mean and covariance structure of the Gaussian model.


\emph{Priors that avoid overfitting and lead to a robust estimation}: Fig. \ref{fig:overfit} shows posterior means of the leptokurtosis parameter $\eta$ of the NIG distribution when fitting i.i.d. standard Gaussian data. These were obtained by numerical integration of the unnormalized $\pi(\eta|\mathbf{y}) \propto \prod_i^n \pi_{\text{NIG}}(y_i|\eta) \pi(\eta)$, where $\pi_{\text{NIG}}$ is the probability density function (PDF) of the standardized and symmetric NIG distribution of subsection \ref{sect:varpar}.  The Gaussian and Cauchy distributions are limit distributions when $\eta \to 0$ and $\eta \to \infty$, respectively. If we use a uniform prior for small samples the results suggest the necessity of a non-Gaussian model when the simpler Gaussian model should be the preferred one. We should prefer the Gaussian model not only for parsimony's sake but also because for Gaussian simulated data, the expected predictive performance will be the highest for the Gaussian model if a strictly proper scoring rule is used to evaluate the predictive accuracy \citep{gneiting2007strictly}. The tendency to overestimate the degree of non-Gaussianity in the data is more severe when a non-Gaussian model is added as a latent component in a hierarchical model and can be prevalent even for samples sizes of 500, as our simulation study in section \ref{sect:sim} shows. This issue motivates an inferential framework that contracts non-Gaussian models towards Gaussianity in the absence of sufficient evidence of non-trivial leptokurtosis or asymmetry. In section \ref{sect:priors} we construct prior distributions for the skewness and leptokurtosis parameters of the NIG and GAL distribution based on the penalized complexity (PC) priors principled approach of \cite{simpson2017penalising}.

\begin{figure}[]
    \centering
  \includegraphics[width=0.6\linewidth]{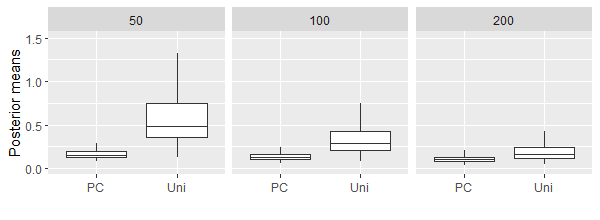} 
  \includegraphics[width=0.39\linewidth]{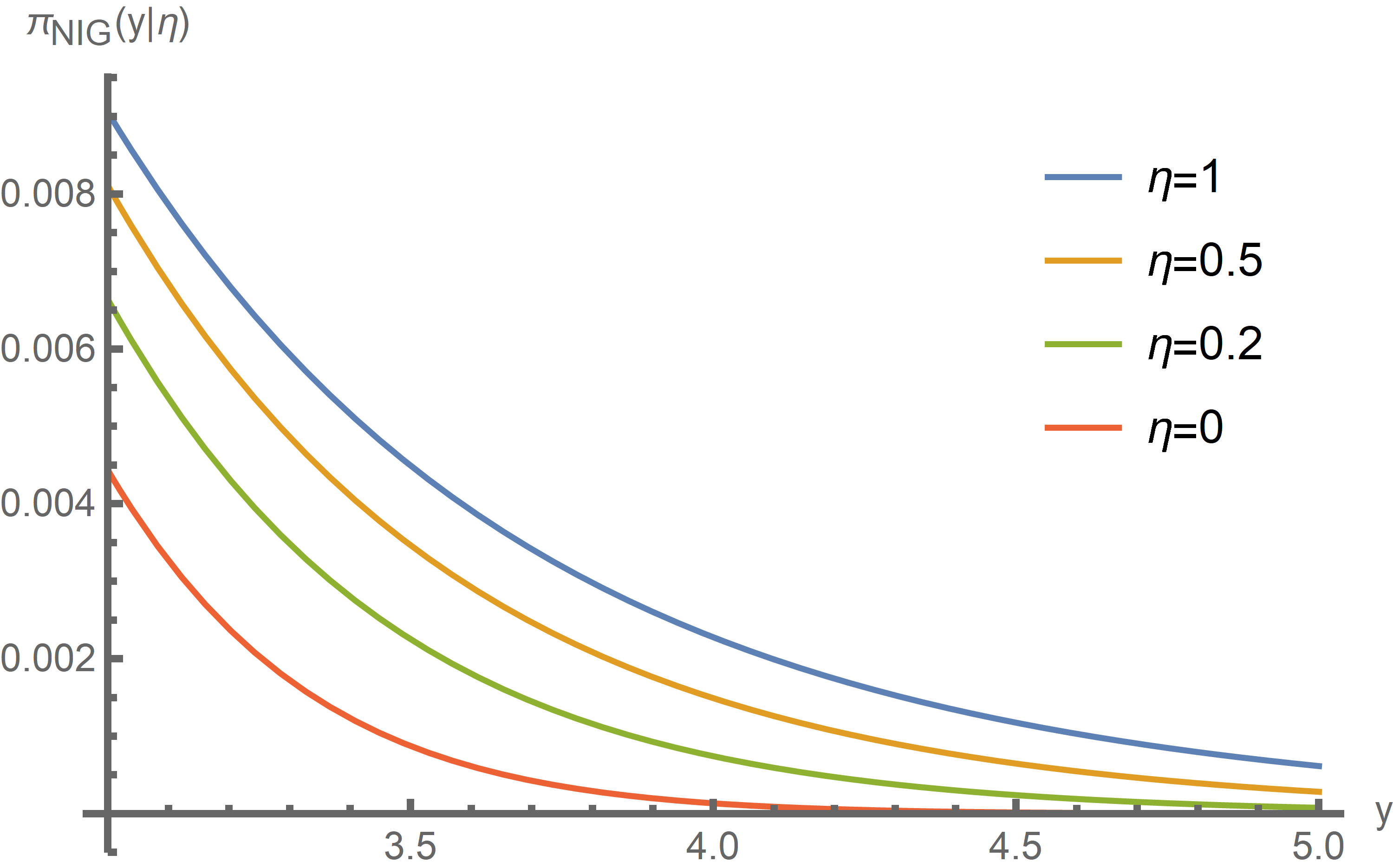} 
  \caption{The box plots show the posterior means of the leptokurtosis parameter $\eta$ for i.i.d. Gaussian simulated data with sample sizes of $n=50,100,200$ with penalized complexity (PC) and uniform (Uni) priors, repeated over 1000 experiments. The right plot shows the right tail of the symmetric NIG PDF for several values of $\eta$.}
  \label{fig:overfit}
\end{figure}

\emph{Many models, one framework}:  We presented in the previous subsection a wide range of models, defined in discrete or continuous space, for time series or spatial data. We unify these models into a single generic class, which we examine in section \ref{sect:distributions}. We also offer in section \ref{sect:manymodels} a set functions for the $Stan$ platform \citep{stan} that allow a straightforward implementation of these models.

\emph{Bayesian analysis}: A simulation study is conducted to compare the PC priors with other priors choices in section \ref{sect:sim}, and the developed methods are applied to two spatial datasets in section \ref{sect:appli}. The analysis showed that the PC priors achieve the sought contraction towards the Gaussian model when there is not enough convincing evidence in the data of non-Gaussianity and lead to more robust estimation.

Finally, section \ref{sect:discussion} contains a summary and discussion of future work and possible extensions.

\section{A flexible extension of Gaussian models}\label{sect:distributions}

The GH distribution \citep{barndorff1978hyperbolic} can be conveniently represented as a variance mixture of normal distributions, where the mixing distribution is a generalized inverse Gaussian (GIG) random variable. If a random variable $\Lambda$ follows the GH distribution, then it has the following hierarchical representation 
\begin{align} \begin{split} \label{eq:noise1}
\Lambda|V &\sim \text{N}(\mu  + \beta V, V),  \ \ V  \sim \text{GIG}(\lambda, \delta, \sqrt{\alpha^2-\beta^2}),
\end{split} \end{align}
 where $\delta$,$\mu$, and $\beta$ are scale, location, and skewness parameters, respectively, while $\alpha$ and $\lambda$ are two shape parameters. The constraints on these parameters are found in \cite{barndorff2012levy}. The GH distribution includes many widely used distributions as special cases, including the Gaussian, t-Student, Cauchy, NIG, and GAL distributions. The NIG subclass is obtained by setting $\lambda=-1/2$, leading to an inverse Gaussian (IG) mixing distribution. The PDF of a NIG distribution is
$$
f(x; \alpha, \beta, \delta, \mu)=\frac{\alpha \delta}{\pi} \frac{K_{1}\left(\alpha \sqrt{\delta^{2}+(x-\mu)^{2}}\right)}{\sqrt{\delta^{2}+(x-\mu)^{2}}} e^{\delta \sqrt{\alpha^{2}-\beta^{2}}+\beta(x-\mu)},
$$
where $K_\lambda(x)$ is the modified Bessel function of the second kind of order $\lambda$. On the other hand, the GAL subclass is obtained from the limit $\delta \to 0$, leading to a gamma mixing variable and the PDF
$$
f(x; \lambda, \alpha, \beta, \mu)=\frac{\left(\alpha^{2}-\beta^{2}\right)^{\lambda}}{\sqrt{\pi}(2 \alpha)^{\lambda-\frac{1}{2}} \Gamma(\lambda)}|x-\mu|^{\lambda-\frac{1}{2}} K_{\lambda-\frac{1}{2}}(\alpha|x-\mu|) e^{\beta(x-\mu)}.
$$

As a necessary step in building explicable priors, in the following subsections \ref{sect:varpar} and \ref{sect:tailpar} we present a standardized parameterization of the NIG and GAL distributions which also aims to achieve an orthogonal interpretation of the parameters. Then in subsection \ref{sect:framework}, we define the distribution of multivariate models that are driven by NIG or GAL noise.

\subsection{Standardized parameterization  $(\eta, \zeta)$} \label{sect:varpar}

In the previous parameterization $(\lambda, \alpha, \beta, \delta, \mu)$, the location ($\mu$) and scale ($\delta$) parameters do not correspond to the mean and standard deviation of the distribution. For the NIG distribution setting $\mu=0$ and $\delta=1$ leads to $E[\Lambda] = \beta/(-\alpha^2+\beta^2)$ and \mbox{$V[\Lambda]=\alpha^2/(\alpha^2-\beta^2)^{3/2}$}, and so the mean and variance depend on the degree of asymmetry and leptokurtosis of the distribution. There are alternative parameterizations of the GH distribution in the literature \citep{paolella2007intermediate, prause1999generalized}, including location-scale invariant ones, but these continue to not be property-based.

As in \cite{niekerkskewprobit} we posit that when constructing priors the mean and the standard deviation should be fixed, instead of location and scale parameters. This not only eases interpretation but also allows assigning the priors for the mean and scale of the Gaussian model in non-Gaussian models, which is very convenient when implementing these non-Gaussian models in practice (otherwise reasonable priors for the location and scale would have to depended on $\alpha$ and $\beta$). Table \ref{tab:GH} shows the conversion between the standardized parameterization and the more conventional one in eq. \eqref{eq:noise1}, along with some relevant moments (consider for now $h=1$). The variance-mean mixture representation belonging to this parameterization is
\begin{align} \label{eq:noise21}
\Lambda|V &\sim \text{N}\left(\frac{1}{\sqrt{1+\eta \zeta^2}} \zeta (V-h), \ \ \frac{1}{1+\eta \zeta^2}V\right),
V {\sim}  \begin{cases}
      \text{IG}(1,\eta^{-1})   &\text{(NIG)} \\
      \text{Gamma}(\eta^{-1},\eta^{-1})  &\text{(GAL).}
    \end{cases} 
\end{align}

A location and scale parameter can be added by the usual transformation $m + \sigma\Lambda$, and the mean and variance will be $m$ and $\sigma^2$, regardless of the other parameters' value. The parameter $\eta$ is related to the degree of non-Gaussianity since $\Lambda$ converges to a Gaussian random variable when $\eta \to 0$, and as we increase $\eta$, the excess kurtosis increase.  The parameter $\zeta$ is related to the asymmetry of the random variable since $\Lambda$ is symmetric when $\zeta=0$, and when $\zeta>0$, it is skewed to the right. Here and henceforth, we will make use of the notions of ``base model'', ``flexible model'' and ``flexibility parameter'' defined in \cite{simpson2017penalising}. We see the non-Gaussian model parameterized by $(\eta, \zeta)$ as a flexible extension of the base Gaussian model since it contains the Gaussian model as a special case (when $\eta=\zeta=0$) and deviations from the Gaussian model are quantified by the parameters $\eta$ and $\zeta$ to which we refer as flexibility parameters. The densities of the NIG and GAL distributions with the previous parameterization are displayed in Fig. \ref{fig:NIGplot}. We observe that the larger the value of $\eta$, the larger the asymmetry induced by the same value of the parameter $\zeta$.


\begin{figure}[h]
    \centering
 \includegraphics[width=\linewidth]{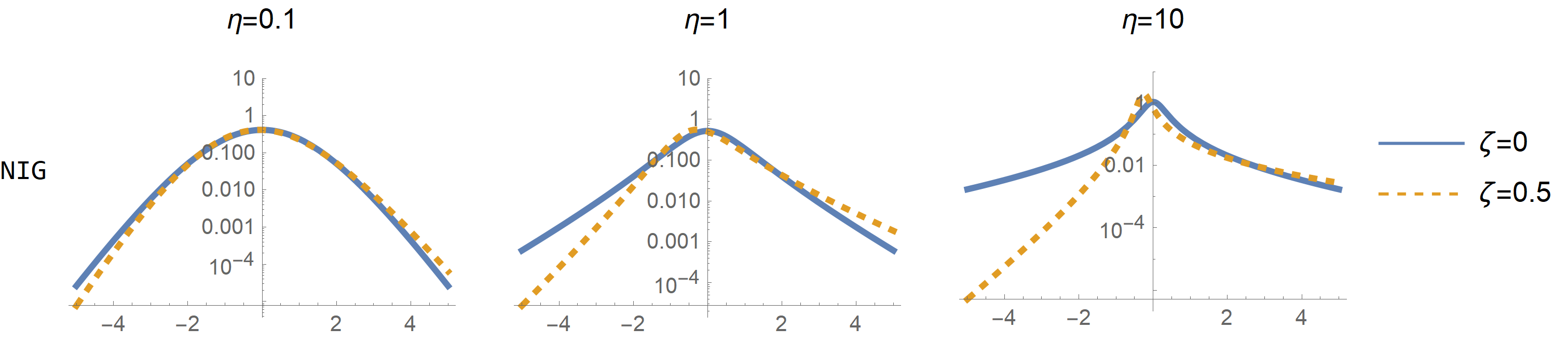} 
 \includegraphics[width=\linewidth]{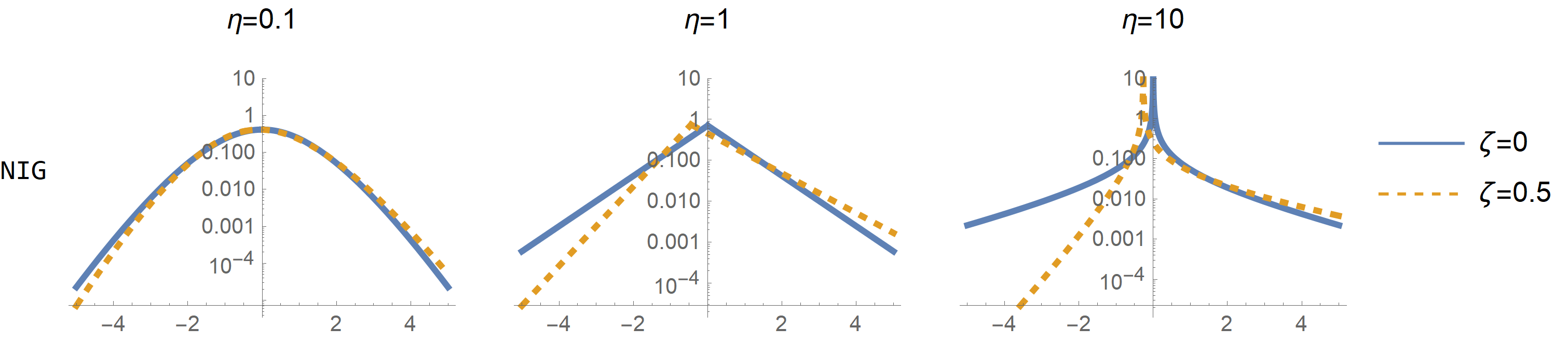} 
  \caption{Density of the NIG (top) and GAL (bottom) distributions in log scale for different values of $\eta$ and $\zeta$ with $\sigma=1$ and $h=1$.}
  \label{fig:NIGplot}
\end{figure}

\begin{table}
    \centering
    \begin{tabular}{cccccccccc}
    \hline
    \textbf{Dist.} & $\boldsymbol{\lambda}$ & $\boldsymbol{\alpha}$                               & $\boldsymbol{\beta}$           & $\boldsymbol{\delta}$          & $\boldsymbol{\mu}$  & \textbf{S}                            & \textbf{EK}               \\ \hline
    NIG            & -1/2             & $\frac{1}{\tilde{\sigma}} \sqrt{\frac{1}{\eta}+\zeta^2}$  & $\frac{\zeta}{\tilde{\sigma}}$ & $\frac{\tilde{\sigma}}{\sqrt{\eta}}h$ & $-\tilde{\sigma}\zeta h$     & $\frac{3\zeta\eta}{\sqrt{h+h\zeta^2\eta}}$             & $\frac{3\eta(1+5\zeta^2\eta)}{h(1+\zeta^2\eta)}$       \\
    GAL            & $h\eta^{-1}$        & $\frac{1}{\tilde{\sigma}} \sqrt{\frac{2}{\eta}+\zeta^2}$ & $\frac{\zeta}{\tilde{\sigma}}$ & 0                        & $-\tilde{\sigma}\zeta h$  & $\frac{\zeta\eta(3+2\zeta^2\eta)}{\sqrt{h(1+\zeta^2\eta)^{3}}}$ & $\frac{3 \eta (1 + 4 \zeta^2 \eta + 2 \zeta^4 \eta^2)}{h(1 + \zeta^2 \eta)^2}$ \\ \hline
    \end{tabular}
    \caption{Parameters of the GH distribution, skewness (S), excess kurtosis (EK) for the GAL and NIG special cases, where $\tilde{\sigma} =  1 / (\sqrt{1+\eta \zeta^2})$. }
    \label{tab:GH}
\end{table}

\subsection{Standardized and orthogonal parameterization $(\eta^\star,\zeta^\star)$}
\label{sect:tailpar}

There is still some confounding between $\eta$ and $\zeta$ in the standardized parameterization since the excess kurtosis increases with $\zeta$, so a prior for $\eta$ derived for a symmetric model may not shrink the models towards Gaussianity as much as we anticipate if the data is asymmetric. Thus, we need to associate $\eta$ with some interpretable model property and then guarantee that the property is invariant with the second flexibility parameter. This will also allow using the PC prior of $\eta$, which is derived in section \ref{sect:priors} for the symmetric case, in asymmetric data. 

\sloppy We find the kurtosis to be hard to interpret (it is not clear what an increase in kurtosis of 1 means in practice) and so we prefer to associate $\eta$ with the likelihood of large events. For the NIG distribution this probability can be approximated by \mbox{$P(|\Lambda|>x) \propto x^{-3/2} \exp(-\eta^{-1/2} \xi_{NIG} x)$} for large $x$ \citep{hammerstein2016tail}. The dependency of this probability with the skewness parameter $\zeta$ comes through the rate $\xi_{NIG}$:
$$\xi_{\text{NIG}} = 1+\zeta^2\eta - |\zeta|\sqrt{\eta(1+\zeta^2\eta)},$$
which is equal to 1 in the symmetric case ($\zeta=0$). We require the probability $P(|\Lambda|>x)$ to be invariant with the skewness parameter, at least for large $x$. This can be achieved by the parameter transformations $\eta^\star = \eta\xi_{NIG}^{-2}$ and $\zeta^\star = \zeta\sqrt{\eta}$. The same transformation applies for the GAL distribution where one should take $$\xi_{\text{GAL}}=\sqrt{1+\zeta^2\eta}(\sqrt{2+\zeta^2\eta}-|\zeta|\sqrt{\eta}).$$

We note that these transformations on the flexibility parameters do not affect the mean (0) and the variance (1) of the standardized parameterization. The quantity $\xi$ is also related to the excess kurtosis, namely $\text{EKurt}[\Lambda]\approx 9/\xi^2$, for both the NIG and GAL distributions. Therefore, by guaranteeing that $P(|\Lambda|>x)$ is invariant with $\zeta^\star$ for large $x$, we also guarantee that the kurtosis of the noise (and the process' marginal distributions) remains approximately the same for different values of $\zeta^\star$, leading to a more orthogonal interpretation of the parameters.
\subsection{Multivariate models driven by non-Gaussian noise}\label{sect:framework}


 
Let $\mathbf{x}^G$ be a random vector that follows a multivariate Gaussian distribution with dimension $n$, mean $m$, and precision matrix $\mathbf{Q}= \sigma^{-2} \mathbf{\mathbf{D}}^T \mathbf{D}$. It can be expressed through \begin{equation}\label{eq:gaussian}
\mathbf{D}(\mathbf{x}^G-m) \overset{d}{=} \sigma \mathbf{Z},
\end{equation}
where $\mathbf{Z} = [Z_1, \dotsc, Z_n]^T$ is a vector of i.i.d. standard Gaussian variables. When $\mathbf{x}^G$ is derived from a process defined in continuous space, one usually has $Z_i  \overset{ind.}{\sim} \text{N}(0,h_i)$ for some predefined constant $h_i$ (for instance, the distance between locations in an autoregressive process), and so the precision matrix is $\mathbf{Q}= \sigma^{-2} \mathbf{D}^T\text{diag}(\mathbf{h})^{-1}\mathbf{D}$. The non-Gaussian extension for $\mathbf{x}$ consists in replacing the driving noise distribution:
\begin{equation}\label{eq:frame}
\mathbf{D}(\mathbf{x}-m)\overset{d}{=} \sigma\mathbf{\Lambda},
\end{equation}
where $\boldsymbol{\Lambda} = [\Lambda_1, \dotsc, \Lambda_n]^T$ is a vector of independent and standardized NIG or GAL random variables that depend on parameters $\eta$ and $\zeta$. Using the parameterizations in subsections \ref{sect:varpar} and \ref{sect:tailpar} for $\mathbf{\Lambda}$, the non-Gaussian random vector $\mathbf{x}$ has the same mean and covariance matrix as the Gaussian random vector $\mathbf{x}^G$ but with more flexible sample path properties and marginal distributions whose kurtosis and skewness are regulated by the flexibility parameters $\eta$ and $\zeta$. Appendix B contains more properties about non-Gaussian models defined via eq. \eqref{eq:frame}.  It is important to mention that these models are not uniquely specified by their covariance or precision matrices but through the matrix $\mathbf{D}$. In section \ref{sect:manymodels} we show that many of the non-Gaussian models mentioned in the Introduction belong to the generic class of models defined by eq. \eqref{eq:frame}.

A variance-mean mixture representation of the non-Gaussian random vector  $\mathbf{x}$ is obtained by considering the random vector $\mathbf{V}= [V_1, \dotsc, V_n]^T$ and the predefined vector $\mathbf{h}= [h_1, \dotsc, h_n]^T$, containing the mixing distributions and some predefined constants, respectively. In the standardized parameterization it is:
\begin{align} \label{eq:model}
\begin{split} 
\mathbf{x}|\mathbf{V} &\sim \text{N}\left(m + \frac{\sigma}{\sqrt{1+\zeta^2\eta}} \zeta \mathbf{D}^{-1}(\mathbf{V}-\mathbf{h}),\frac{\sigma^2}{1+\zeta^2\eta} \mathbf{D}^{-1}\text{diag}(\mathbf{V})\mathbf{D}^{-T} \right), \\
V_{i} &\overset{ind.}{\sim}  \begin{cases}
      \text{IG}(h_i,\eta^{-1} h_i^2)   &\text{(NIG Noise)} \\
      \text{Gamma}(h_i \eta^{-1},\eta^{-1})  &\text{(GAL Noise).}
    \end{cases}
  \end{split}
\end{align}

To broaden the applicability of the multivariate model $\mathbf{x}$ to data measured at irregularly spaced locations or time intervals, we need to consider the predefined vector $\mathbf{h}$. Take, for instance, a continuous Gaussian random walk of order 1 (RW1) process evaluated at locations $x_1, x_2, \dotsc$ with distances $h_i=x_{i+1}-x_{i}$. In this situation, the model is defined by the increments $x_{i+1}-x_{i}$ following a normal distribution with mean 0 and variance $h_i\sigma^2$ \citep{rue2005gaussian}. Therefore, a larger distance between observations will lead to a noise Gaussian distribution with increased variance. A similar transformation applies when we discretize a NIG or GAL continuous process, but the increased distance between locations will not only change the variance to $h_i\sigma^2$ but also will the shape of the distribution. This is a standard result from the theory of Lévy processes, where the new densities are found by raising the NIG and GAL characteristic functions to the power $h_i$ and then taking the inverse-Fourier transform (see Appendix A and \cite{barndorff2012levy}). The variables $\Lambda_i$ will share the $\sigma$, $\eta$ and $\zeta$ parameters, but $h_i$ will generally be unique to each noise variable (see Table \ref{tab:GH}).  If the observations are measured at equally spaced locations we set $\mathbf{h}$ to be a vector of ones.

\section{Contraction towards Gaussianity}\label{sect:priors}

Our goal is to control the flexibility of non-Gaussian models and shrink the model towards Gaussianity if there is insufficient evidence in the data of excess kurtosis and asymmetry. One natural and simple way to achieve this contraction is by adding in the inferential framework priors for the flexibility parameters $\eta$ and $\zeta$ that give preference to the Gaussian model. 

The prior distributions derived next are applicable to a large variety of models and can be easily added to pre-existing model implementations.  We will make use of the penalized complexity (PC) priors methodology of \cite{simpson2017penalising}. The approach consists of setting an exponential distribution on the unidirectional \emph{distance} $d(\cdot)$, measuring the added complexity of the more flexible model, with respect to the base model. The measure of complexity is based on the Kullbeck-Leibler divergence (KLD), and $d(\eta)=\sqrt{2KLD(\pi(\mathbf{x}|\eta) \ || \ \pi (\mathbf{x}|\eta=0))}$.  PC priors give preference to the simpler model and tend to avoid overfitting by default because the mode of the prior distribution is located at the base (Gaussian) model: $\pi(\eta) = \theta |d'(\eta)| \exp{(-\theta d(\eta))}$. The previous prior is defined up to a rate parameter $\theta$ which needs to be calibrated (see subsection \ref{sect:scalingeta}).

\begin{remark}
An important feature of the class of non-Gaussian models that we are studying, which are obtained by replacing Gaussian white noise with more flexible alternatives, is that from an information-theoretic perspective, the deviation from the Gaussian model only depends on the flexible noise and does not depend on the covariance structure encoded by the matrix $\mathbf{D}$. This may not seem obvious at first, since different choices for the matrix $\mathbf{D}$ lead to different sample path behaviors and marginal properties (see \mbox{Fig. \ref{fig:compare1}}). However, if we construe the density of the non-Gaussian random vector as linearly transformed NIG or GAL densities $\mathbf{x} = m + \sigma\mathbf{D}^{-1} \mathbf{\Lambda}$, and similarly for the Gaussian random vector $\mathbf{x}^G =m + \sigma \mathbf{D}^{-1} \mathbf{Z}$, this property follows directly from the well know invariance of the KLD under monotonic transformations:
\begin{equation*}\label{eq:KLDinv}
KLD(\mathbf{x} \ || \ \mathbf{x}^G) = KLD(m + \sigma\mathbf{D}^{-1}\boldsymbol{\Lambda} \ || \ m + \sigma\mathbf{D}^{-1} \boldsymbol{Z}) = KLD(\boldsymbol{\Lambda} \ ||\ \boldsymbol{Z}) = \sum_{i=1}^n KLD(\Lambda_i \ || \ Z_i). 
\end{equation*}
\end{remark}
\vspace{-0.5cm}
The previous property implies that the distribution of the PC prior does not depend on $\mathbf{D}$, $\sigma$, or $m$ and is therefore applicable to any non-Gaussian model that can be expressed via eq. \eqref{eq:frame}.  Thus, without loss of generality, we consider in this section processes with 0 mean ($m=0$) and driving noise with unit variance ($\sigma=1$). In the following subsection, we derive the PC prior of $\eta$ for the symmetric case ($\zeta=0$), and then in subsection \ref{Sect:PCpriormu}, we find the PC prior of $\zeta$ conditioned on $\eta$.

\subsection{PC prior distribution for the first flexibility parameter}
\label{Sect:PCprioreta}

The non-Gaussian extension presented in subsection \ref{sect:framework} preserves the mean and covariance structure of $\mathbf{x}$, so we can assume that both the flexible and base models have the same scale parameter $\sigma$ and the same spatial (or temporal) range parameter, such as the parameter $\kappa$ in the Matérn model of eq. \eqref{eq:SPDE}. These parameters cancel out when computing $KLD(\mathbf{x} \ || \ \mathbf{x}^G)$ as seen in Remark 1, and therefore $KLD(\mathbf{x} \ || \ \mathbf{x}^G)$  only depends on $\eta$ and $\mathbf{h}$.  The non-Gaussian noise $\Lambda_i$ can follow either a NIG or a GAL distribution, and for both cases, $KLD(\Lambda_i \ || \ Z_i)$, which should be seen as a function of $\eta$, behaves quadratically near the base model (see \mbox{Fig. \ref{fig:kld} (a)}). In the symmetric case the parameterizations $(\eta, \zeta = 0)$, and $(\eta^\star, \zeta^\star = 0)$ are equivalent, therefore the same prior distribution is assigned for $\eta$ and $\eta^\star$.

\begin{figure}[]
     \centering
    \includegraphics[width=0.49\linewidth]{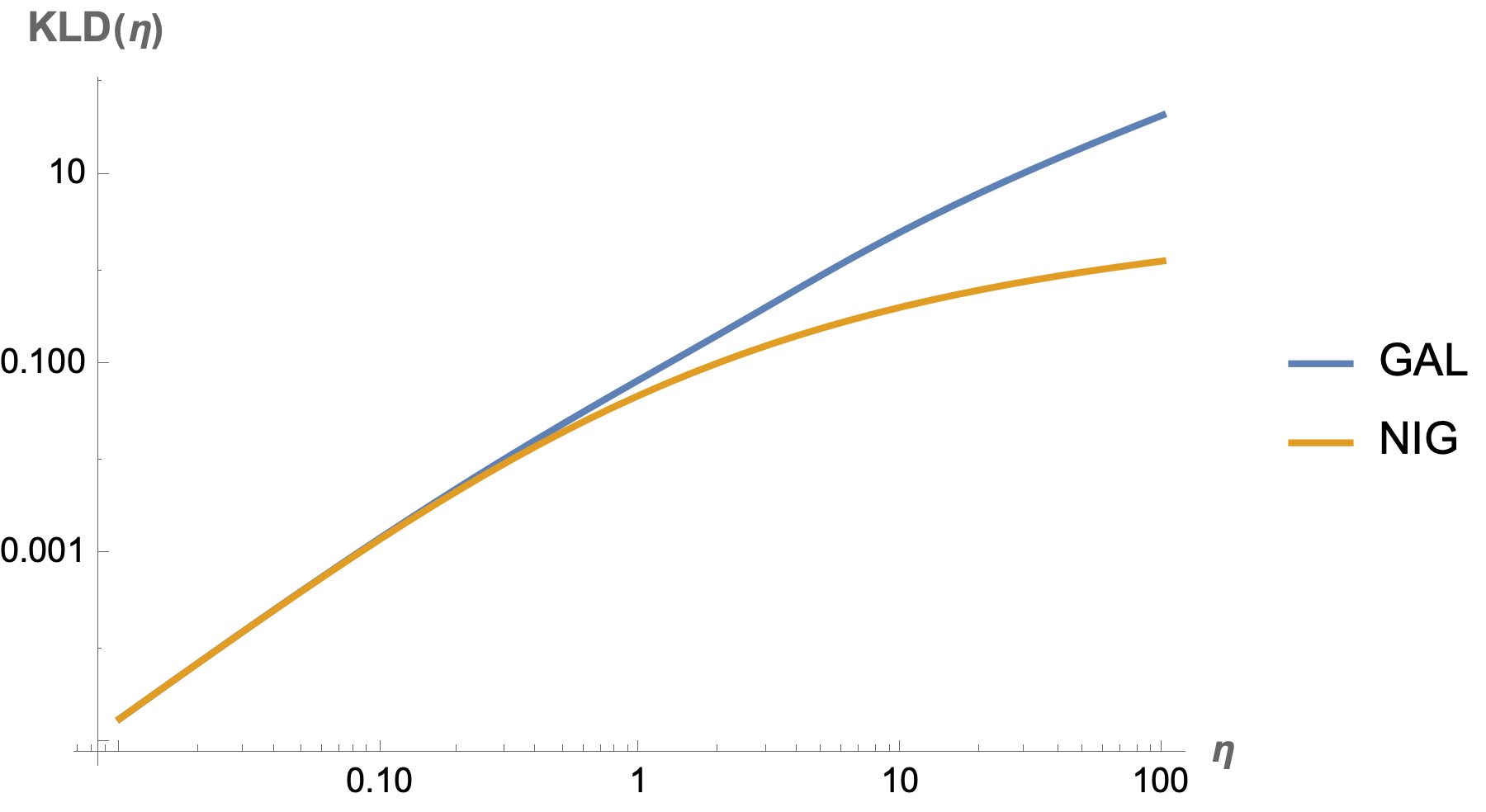}    \includegraphics[width=0.49\linewidth]{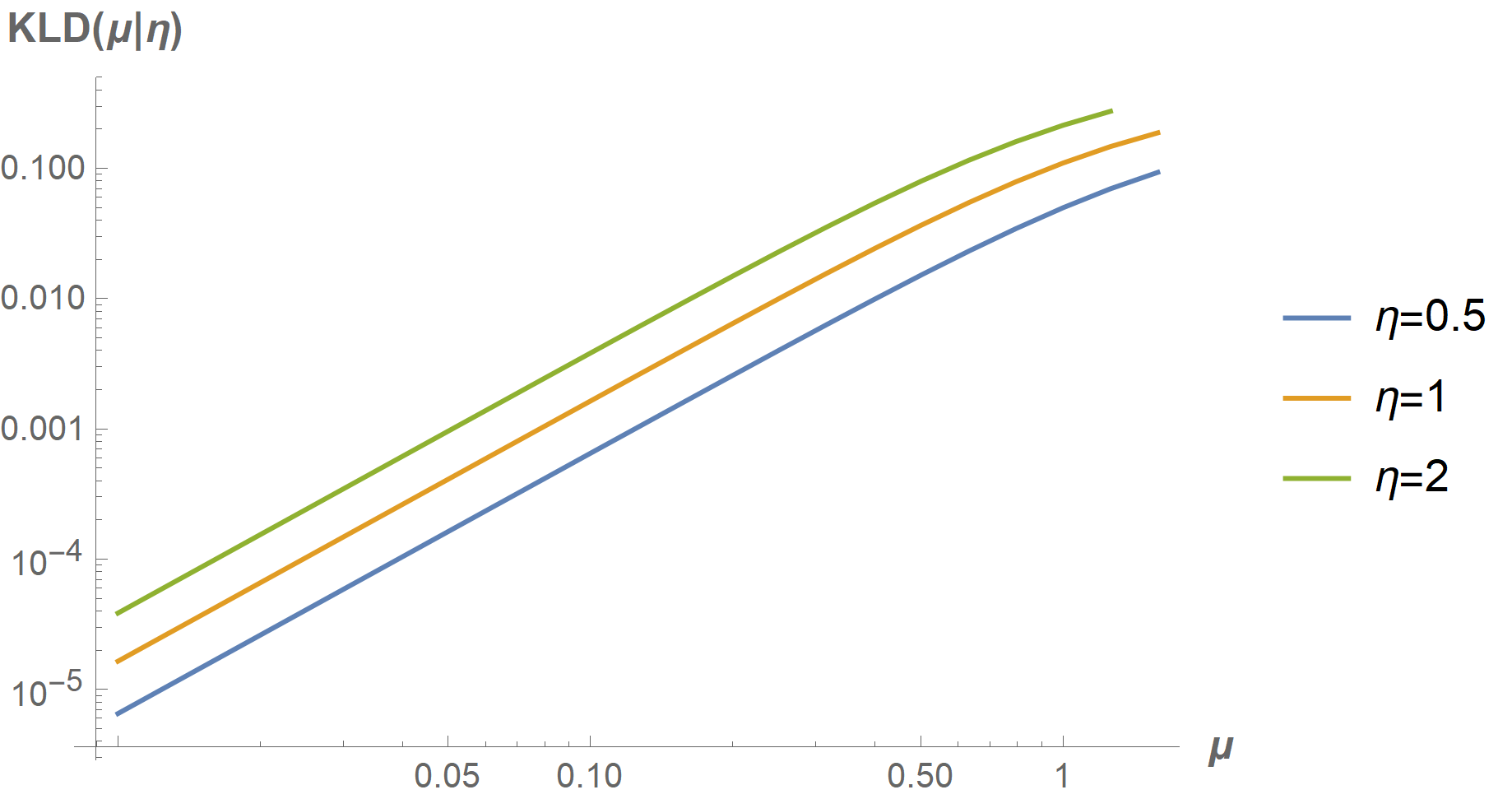}
    \caption{$KLD$ computed by numerical integration for $h=1$ in log-log scale. The plot on the left shows $KLD(\eta,\zeta=0)$ for NIG and GAL noise, and the plot on the right shows $KLD(\zeta|\eta)$ for three fixed values of $\eta$ for NIG noise.}
    \label{fig:kld}
\end{figure}
\begin{theorem}\label{theo:1}
Let $\mathbf{D}$ be a $n \times n$ non-singular matrix. Also, let the flexible model $\mathbf{x}$ with density $\pi(\mathbf{x}|\eta, \zeta, h)$ be defined by $\mathbf{D}\mathbf{x}=\mathbf{\Lambda}$, where $\mathbf{\Lambda}$ is a vector of independent NIG or GAL noise defined in subsection \ref{sect:varpar}. Then, for small $\eta$, the KLD is
$$KLD\left( \ \pi(\mathbf{x}|\eta, \zeta=0) \ || \ \pi^G(\mathbf{x}|\eta=0, \zeta=0) \ \right) = \frac{3}{16}\left(\sum_{i=1}^n\frac{1}{h_i^2}\right)\eta^2 + \mathcal{O}(\eta^4).$$
where $\pi^G(\mathbf{x}|\eta=0,\zeta=0)$ is the density of the base Gaussian model.
\end{theorem}
\begin{proof}
See Appendix D.
\end{proof}

We are mostly interested in penalizing the added complexity of the more flexible non-Gaussian model in a neighborhood around the base Gaussian model, and as suggested by \cite{simpson2017penalising}, a Taylor expansion around the base model is done, and the second order expansion is used as the measure of added complexity.

\begin{corollary}
The distance measure is $d(\eta) = \sqrt{2 KLD(\eta)} \propto \eta$ and so the PC prior for $\eta$ and $\eta^\star$ will follow an exponential distribution with some rate parameter $\theta_\eta$. We note that this PC prior induces a LASSO (L-1) style penalty on the parameter $\eta$ (and $\eta^\star$), since $\log \pi(\eta) = -\theta_\eta \eta + \text{const}$.  
\end{corollary}


\subsection{PC prior distribution for the second flexibility parameter} \label{Sect:PCpriormu}

We derive in this subsection the PC prior for the second flexibility parameter. As seen in section \ref{sect:varpar}, the impact of the parameter $\zeta$ on the the NIG and GAL distributions depend on the value of $\eta$ (if $\eta=0$, then $\zeta$ has no impact).  Therefore, we derive the PC prior for $\zeta$ conditionally on $\eta$, and consider that the base model is a non-Gaussian model driven by symmetric noise $\mathbf{\Lambda}^{Sym}$ (with parameters $\eta$ and $\zeta=0$), and the flexible model $\mathbf{x}$ is driven by asymmetric noise $\mathbf{\Lambda}$ (with parameters $\eta$ and $\zeta$). In Fig. \ref{fig:kld} (b), we show the KLD between the NIG noises $\Lambda_i$ and $\Lambda_i^{Sym}$ for several values of $\eta$ and $\zeta$. 


A Taylor expansion around $\eta=0$ and $\zeta=0$ yields a quadratic dependency with $\eta$ which is only accurate when $\eta$ is very close to 0: $$KLD(\Lambda_i \ || \ \Lambda_i^{Sym}) = ( 3\eta^2/(4h_i) + \mathcal{O}(\eta^3) )\zeta^2 +  \mathcal{O}(\zeta^4).$$ Therefore, we use the following upper bound of the KLD as a measure of added complexity, which provides a more reasonable approximation. 

\begin{theorem}\label{theo:2}
Under the same conditions of Theorem \ref{theo:1}, we have for NIG driving noise
$$KLD( \ \pi(\mathbf{x}|\eta, \zeta) \ || \ \pi^{Sym}(\mathbf{x}|\eta, \zeta=0) \ ) \leq \frac{n}{2}\eta\zeta^2. $$
This inequality also holds for GAL driving noise when $\eta < \min_{i=1,\dotsc,n} h_i$.
\end{theorem}
\begin{proof}
See Appendix D.
\end{proof}

If we use the previous upper bound as a measure of complexity, the distance $d(\zeta)$ is  $\sqrt{2 KLD(\zeta)} \propto \sqrt{\eta}|\zeta|$ and by setting an exponential distribution on $d(\zeta)$ with rate parameter $\theta_{\zeta}$, the density of $\zeta|\eta$ is found to be
\begin{equation}\label{eq:mudensity}
 \pi(\zeta|\eta) = \frac{1}{2} \theta_\zeta \sqrt{\eta} \exp\left(-\theta_\zeta \sqrt{\eta} |\zeta| \right).
\end{equation}
\begin{corollary}
With the standardized and orthogonal parameterization $(\eta^\star, \zeta^\star)$, the skewness parameter is $\zeta^\star = \sqrt{\eta} \zeta$ and so $KLD \leq n{\zeta^\star}^2/2$. We then have $d(\zeta) \propto |\zeta^\star|$ and the PC prior for $\zeta^\star$ is a Laplace distribution with rate parameter $\theta_\zeta$. This prior also acts as a LASSO (L-1) penalty on the estimation of $\zeta^\star$, since $\log \pi(\zeta^\star)  = -\theta_\zeta |\zeta^\star| + \text{const}$.
\end{corollary}

\subsection{Calibration of the PC priors} \label{sect:scalingeta}

In line with the \emph{weakly informative} prior framework of \cite{simpson2017penalising} and \cite{gelman2017prior}, the calibration of the PC priors, i.e, the choice of $\theta_\eta$ and $\theta_\zeta$, is based on the user defining the prior probabilities $P(\eta^\star>U_\eta) = \alpha_\eta$ and $P(|\zeta^\star|>U_\zeta) = \alpha_\zeta$ and in turn the calibration parameters are $\theta_\eta = -\log(\alpha_\eta)/U_\eta$ and \mbox{$\theta_\zeta = -\log(\alpha_\zeta)/U_\zeta$}. Of course, sensible choices for the upper-tail values $U$ and probabilities $\alpha$ must come from an understating of how different values of $\eta^\star$ and $\zeta^\star$ impact the process. This understanding can be informed by, for instance, plotting sample paths of the process for increasing values of $\eta^\star$ (see Fig. \ref{fig:compare1}), stopping when we observe unreasonably large spikes or jumps, at $\eta^\star=U_\eta$, and then setting a low probability that $\eta_\star>U_\eta$, say equal to 0.01. A similar procedure can be applied when calibrating the PC prior of $\zeta^\star$, but one should look for asymmetries in the number of large jumps or spikes, for instance for positive $\zeta^\star$, there should be more positive jumps than negative jumps.

More interpretable distribution features can be used in the calibration, such as the probability that large marginal events occur (larger than 3 times the marginal standard deviation, for instance) which increase with $\eta$. However, deriving these probabilities for new models can be unhandy, since the marginal distributions do not have a closed-form PDFs or CDFs. Nevertheless, we pursue this path in Appendix E to calibrate the PC prior of $\eta$ for Matérn and OU processes.

\subsection{Comparison with other prior distributions}

For the GHSST distribution (subclass of the GH distribution) seen in subsection \ref{sect:litreview}, the flexibility parameters are $\nu$ and $\beta$, the base model corresponds to $\nu\to\infty$ and $\beta=0$ and gamma and normal priors are commonly chosen, respectively. In our parameterization, this would suggest an inverse gamma distribution for $\eta$ since the base model is at $\eta \to 0$, and a normal prior for $\zeta$. An inverse gamma prior decays slower to 0 as $\eta$ increases (see Fig. \ref{fig:sim5}), and it has no mass at the base model ($\eta=0$), so it acts as a repellent from the simpler Gaussian model, having the opposite effect that we necessitate. Also, a Gaussian prior for $\zeta$ may not achieve as much contraction as the prior in eq. \eqref{eq:mudensity}, which follows a Laplace distribution for a fixed value of $\eta$, and it does not take into account that the asymmetry induced by a particular value of $\zeta$ increases with $\eta$. 


The Fisher information matrix does not seem to be available in closed form which makes the use of the Jeffreys priors for the flexibility parameters impractical for the class of models defined by eq. \eqref{eq:frame}. Fig. \ref{fig:sim5} shows a numerical approximation of the Jeffreys prior density for $\eta$, based on the univariate pdf of the NIG distribution. It has the mode at the base model, but unlike the PC priors, which decay exponentially, it is almost flat. 

\section{Many models, one framework} \label{sect:manymodels}

Our general model in eq. \eqref{eq:frame} contains a wide variety of non-Gaussian models as special cases that allow departures from Gaussianity within realizations. Here, we review some non-Gaussian models defined in discrete and continuous space and present a set of functions for $Stan$ that allow an easy and efficient implementation of these models.

\subsection{Models defined on discrete space}\label{sect:discreteapp}

A Gaussian autoregressive process (AR) of order 1 assumes that $\sqrt{1-\rho^2}x_1$ and the differences $\{x_{i+1}- \rho x_{i}, \ i > 1\}$ follow independent Gaussian white noise $\text{N}(0,\sigma^2)$ and $|\rho|<1$. If we assume that the differences instead are non-Gaussian white noise, we obtain the linear system $\mathbf{D}_{AR1} \mathbf{x}=\boldsymbol{\Lambda}$, where $\mathbf{D}_{AR1}= (d_{i,j}) \in \mathbb{R}^{n\times n}$ with  $d_{1,1}=\sqrt{1-\rho^2}$, $d_{i,i-1}=-\rho, d_{i,i}=1$ for $i>1$, and the other matrix elements are 0. Likewise, the matrix $\mathbf{D}$ for higher-order AR processes is a lower triangular Toeplitz matrix containing the autoregression coefficients. A vector of autoregressive processes with no intercept also has a representation $\mathbf{D}_{VAR} \mathbf{x}=\boldsymbol{\Lambda}$ when extending the model to non-Gaussianity, by stacking the multivariate time series into a single vector $\mathbf{x}$. To fit areal data \citet{walder2020bayesian} proposed a simultaneously autoregressive (SAR) model driven by non-Gaussian noise that also can be represented by eq. \eqref{eq:frame} for an appropriately specified matrix $\mathbf{D}$.

\subsection{Models defined on continuous space} \label{sect:continuousapp}

A famous class of processes in spatial statistics are stationary Gaussian processes with Matérn covariance functions \citep{MR0169346}. Gaussian processes with this covariance function can be obtained as solutions of the SPDE:
\begin{equation}\label{eq:SPDE}
(\kappa^2 - \Delta)^{\alpha/2} X(\mathbf{s}) = \sigma \mathcal{W}(\mathbf{s}), \ \ \mathbf{s} \in \mathbb{R}^d,
\end{equation}
where $\kappa$ is a spatial range parameter, $\Delta=\sum_i \partial^2/\partial x_i^2$ is the Laplace operator, $\alpha>d/2$ is a smoothness parameter, and $\mathcal{W}(\mathbf{s})$ is a Gaussian white noise process \citep{Whittle}.  The approximation to discrete space in \citep{lindgren2011explicit} uses the finite element method (FEM) to the stochastic weak formulation of the previous SPDE. It begins by expressing the process $X(\mathbf{s})$ as a weighted sum of basis functions, $X(\mathbf{s}) = \sum_{i=1}^{n}w_i\psi_i(\mathbf{s})$, and it was shown that the weights  $\mathbf{w}=[w_1,\dotsc,w_n]^T$ follow the system $\mathbf{D}\mathbf{w}=\mathbf{Z}$, where the predefined vector $\mathbf{h}$ of section \ref{sect:framework} has elements $h_i=\int_{\mathbb{R}^d}\psi(\mathbf{•}{s})d\mathbf{s}$.  The random field $X(\mathbf{s})$ evaluated at locations $\mathbf{s}_1, \mathbf{s}_2, \dotsc$ composes the vector $\mathbf{x}= [X(\mathbf{s}_1), X(\mathbf{s}_2), \dotsc]^T$ and it is given by the linear combination $\mathbf{x}= \mathbf{A}\mathbf{w}$, where $\mathbf{A}$ is the projector matrix with elements $A_{ij}=\psi_i(\mathbf{s}_j)$. Building the matrices $\mathbf{D}$ and $\mathbf{A}$ may seem hard at first since it involves the finite element method, but these can be easily built with the functions \emph{inla.mesh.2D}, \emph{inla.mesh.fem} and \emph{inla.spde.make.A} in the \emph{R} package \emph{INLA}, and the user only needs to supply the location of the observations and some tuning parameters for the discretization mesh \citep{bakka2018spatial}.

 \citet{bolin2014spatial} extended the previous results to Type-G Matérn random fields by replacing the Gaussian noise process $\mathcal{W}(\mathbf{s})$ with a non-Gaussian noise process $\dot{\boldsymbol{\Lambda}}(\mathbf{s})$. As discussed in \citep{wallin2015geostatistical}, for SPDE models, the increments need to be closed under convolution, and so we can only consider the NIG and GAL subclasses of the GH distribution. It was shown that the stochastic weights now follow the system $\mathbf{D}\mathbf{w}=\mathbf{\Lambda}$, where the matrix $\mathbf{D}$ is the same as in the Gaussian case seen before.

Table \ref{tab:differentialD} lists differential operators associated with several stochastic processes and Fig. \ref{fig:compare1} contains their sample paths, obtained by the SPDE approach, considering a Gaussian noise process (left panel) and a NIG noise process (right panel). The characterization of the marginal distributions is in Appendix C.  Matérn and Ornstein–Uhlenbeck (OU) processes are widely used in applications. The continuous random walk model of order 2 (CRW2) is also known as an integrated Brownian motion process, and it has been used by \cite{diggle2014}, and \cite{zhuintrw}.  Whenever the noise takes an extreme value (for instance, near location 0.8), the CRW1 and OU processes will exhibit a distinct jump, and the RW2 and Matérn processes will exhibit a kink (discontinuity in the first derivative).



\begin{table}
    \centering
    \begin{tabular}{c|cccc}
    \hline
Model & CRW1 & OU & CRW2 & Matérn $\alpha=2$ \\  
$\mathcal{D}$ & $\partial_t$ & $\kappa + \partial_t$ & $\partial_t^2$  & $\kappa^2- \partial_t^2$ \\ \hline
    \end{tabular}
    \caption{Differential operators associated with different models.}
    \label{tab:differentialD}
\end{table}



\begin{figure}[h]
    \centering
 \includegraphics[width=\linewidth]{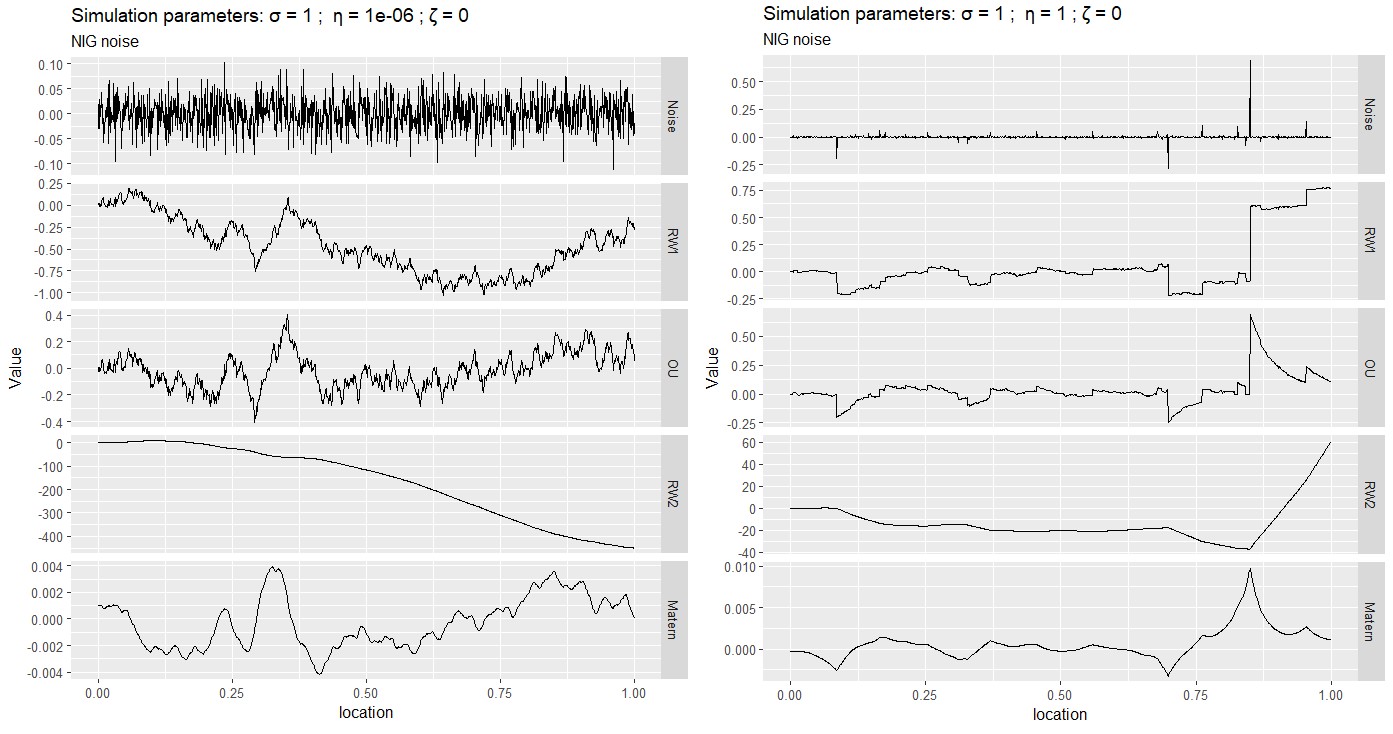} 
  \caption{Noise and sample paths of several models for $\eta=10^{-6}$ (left) and $\eta=1$ (right), for $\sigma=1$ and $\zeta=0$.}
  \label{fig:compare1}
\end{figure}



\subsection{Implementation in Stan}

Estimating non-Gaussian models defined by $\mathbf{D}\mathbf{x}=\mathbf{\Lambda}$ in the $Stan$ platform \citep{stan} can be done by declaring the random vectors $\mathbf{x}$ and $\mathbf{V}$ with the hierarchical representation of eq. \eqref{eq:model}. However, the dimension of $\mathbf{V}$ is the same as the dimension of $\mathbf{x}$, which can be very large in some applications, and since $\mathbf{V}$ needs to be estimated, one can expect long sampling times. We can integrate out the mixing variables $V_i$ in eq. \eqref{eq:model} to reduce the dimension of the parameter space being explored in $Stan$, which can lead to a significant speedup. Note that $\mathbf{x} = \mathbf{D}^{-1}\mathbf{\Lambda}$ and if $\mathbf{D}$ is non-singular, then the multivariate transformation method yields:
\begin{equation} \label{eq:Vintegrated}
    \pi(\mathbf{x})= |\mathbf{D}|\prod_{i=1}^n\pi_{\Lambda_i}([\mathbf{D}\mathbf{x}]_i),
\end{equation}
where $\pi_{\Lambda_i}$ is the PDF of a NIG or GAL distribution. 

The $Stan$ function \emph{nig\_model} returns the log-likelihood of a NIG model based on eq. \eqref{eq:Vintegrated} using the standardized and orthogonal parameterization and it is implemented in \url{github.com/rafaelcabral96/nigstan}. The declaration of $\mathbf{x}$ takes the following form:
\begin{minipage}{\linewidth}
\vspace{0.2cm}
\begin{verbatim}
x     ~ nig_model(D, etas, zetas, h, 1)
etas  ~ exp(theta_eta)
zetas ~ double_exponential(0,1.0/theta_zeta)
\end{verbatim}
\vspace{0.025cm} 
\end{minipage}
where the last argument of \emph{nig\_model} is an integer with value 1 if the log-determinant of $\mathbf{D}$ should be computed (if $\mathbf{D}$ depends on parameters), or 0 otherwise. The last two lines are the declaration of the PC priors for $\eta^{\star}$ and $\zeta^{\star}$.  A location and scale parameter can be added by the usual transformation $m + \sigma\mathbf{x}$. We can fit all models in subsections \ref{sect:discreteapp} and \ref{sect:continuousapp} in $Stan$ using the \emph{nig\_model} function by specifying the appropriate matrix $\mathbf{D}$.  Contrary to the hierarchical representation in eq. \eqref{eq:model}, if we work with eq. \eqref{eq:Vintegrated} there is no need to estimate the auxiliary random vector $\mathbf{V}$ and invert the matrix $\mathbf{D}$, and thus sampling times of hours can be reduced to minutes. The results in the following sections were obtained via this implementation. We demonstrate in more detail how $Stan$ can be used to fit non-Gaussian models for several applications in \url{rafaelcabral96.github.io/nigstan/}. Unfortunately, it is currently not possible to implement models driven by GAL noise in $Stan$ based on eq. \eqref{eq:Vintegrated}, since modified Bessel functions of the second kind with fractional order are currently not available in $Stan$.

\section{Simulation studies} \label{sect:sim}





To study how the priors on the flexibility parameters perform under different conditions, we consider two simulation sets. In the first, we investigate if the posterior distributions of $\eta^\star$ and $\zeta^\star$ are close to the true values of these parameters used to simulate the sample paths. We verify the contraction towards Gaussianity induced by the PC priors and the ability to allow for non-Gaussianity when the latent field is significantly non-Gaussian. 
For the second simulation set, we check how sensitive the posterior distributions of $\eta^\star$ and $\zeta^\star$ are to irregularities in the data for different choices of priors. 

 


\subsection{Implementation details} \label{sect:simdetails}

The following simulation settings were considered. The response is $y_i \sim N(x_i,0.7)$, where $\mathbf{x}$ is non-Gaussian latent field $\mathbf{x}$, defined by eq. \eqref{eq:frame} with mean $m=0$.  The model parameters are $\{\sigma ,\eta^\star, \zeta^\star\}$ and the matrix $\mathbf{D}$ was chosen so that $\mathbf{x}$ corresponds to a Matérn model (with $\alpha=2$ and $\kappa=0.2$) as in subsection \ref{sect:continuousapp}.  We fit the model to simulated data with parameters as shown in Table \ref{table:scenarios}. We considered sample sizes $n$ chosen from $\{50, 100, 500, 1000 \}$ for the first simulation set and $n=500$ for the second simulation set.
 
\begin{table}[b] 
\centering
\begin{tabular}{ccc}
\textbf{Scenarios} & \textbf{Simulation set 1}                                           & \textbf{Simulation set 2}                \\ \hline
Scenario 1         &  $\sigma=1, \eta^\star = \zeta^\star = 0$ (Gaussian)             & Gaussian with no jumps             \\
Scenario 2         & $\sigma=1, \eta^\star = 2, \zeta^\star = 0$ & Gaussian with two jumps of size 25 \\
Scenario 3         & $\sigma=1, \eta^\star = 5, \zeta^\star = 1$           & Gaussian with two jumps of size 50
\end{tabular}
\caption{Parameters of the simulated data for each simulation set and scenario.}
\label{table:scenarios}
\end{table}

    The prior distribution for $\sigma$ is $\text{IGamma}(1,1)$ and we consider several prior configurations for the flexibility parameters $\eta^\star$ and $\zeta^\star$:
  \begin{enumerate*}
    \item[] \textbf{PC1}: $\eta^\star \sim \exp(\theta_\eta=30)$, and $\zeta^\star \sim \text{Laplace}(\theta_\zeta=13)$; 
    \item[] \textbf{PC2}: $\eta^\star \sim \exp(\theta_\eta=2.3)$,  and $\zeta^\star \sim \text{Laplace}(\theta_\zeta=1)$;
    \item[] \textbf{IG1/N1}: $\eta^\star \sim \text{IGamma}(2, 0.1)$,  and $\zeta^\star \sim \text{N}(0, 0.3)$;
    \item[] \textbf{IG2/N2}: $\eta^\star \sim  \text{IGamma}(2, 0.43)$,  and $\zeta^\star \sim \text{N}(0, 1)$; 
    \item[] \textbf{Jeffrey}: $\eta^\star$ follows the Jeffrey prior computed numerically, and $\zeta^\star \sim \text{Uni}(-50,50)$
    \item[] \textbf{Uniform}: $\eta^\star \sim \text{Uni}(0,50)$, and $\zeta^\star \sim \text{Uni}(-50,50)$;  
  \end{enumerate*}

\begin{figure}[]
    \centering
\includegraphics[width=0.49\linewidth]{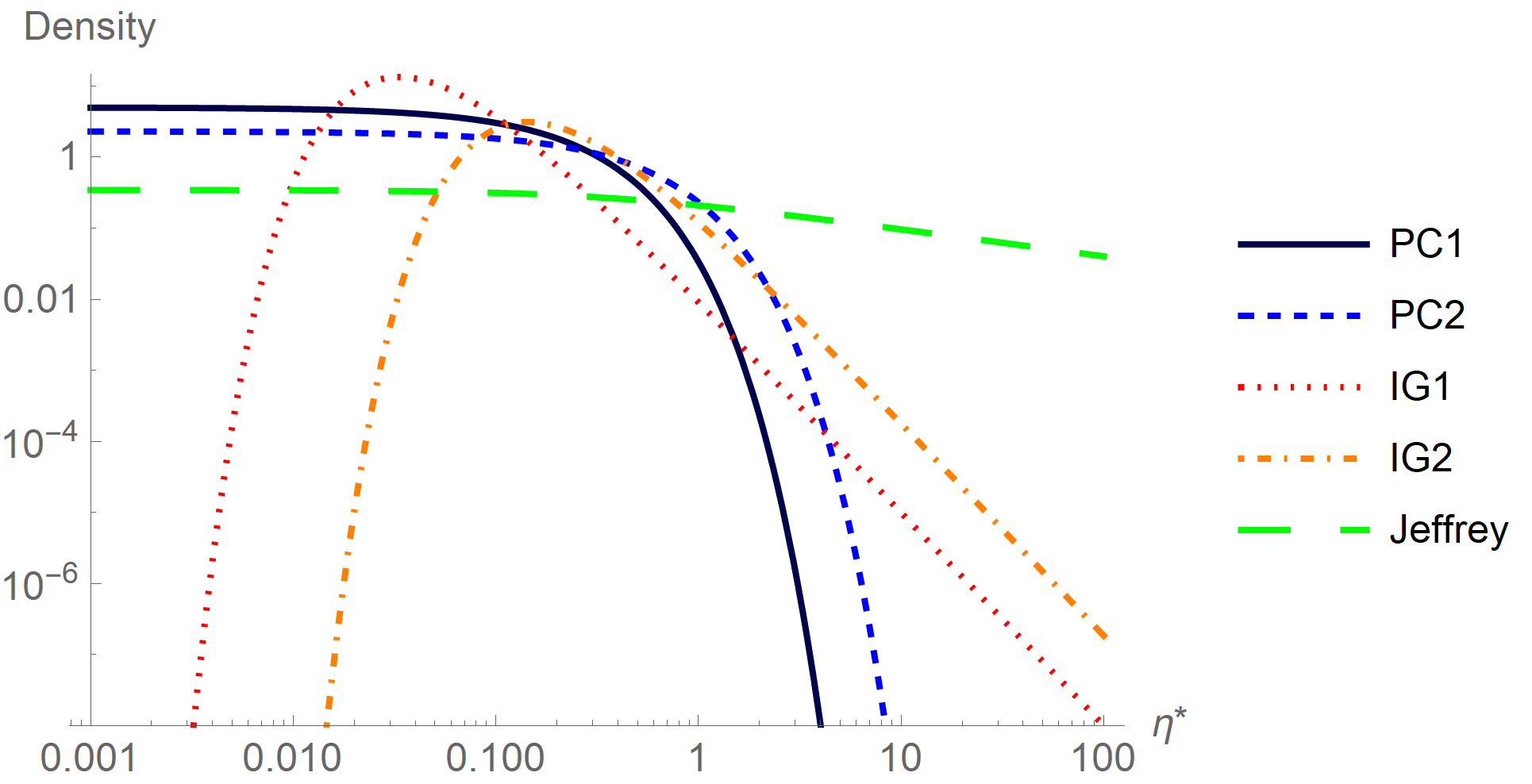}
  \includegraphics[width=0.49\linewidth]{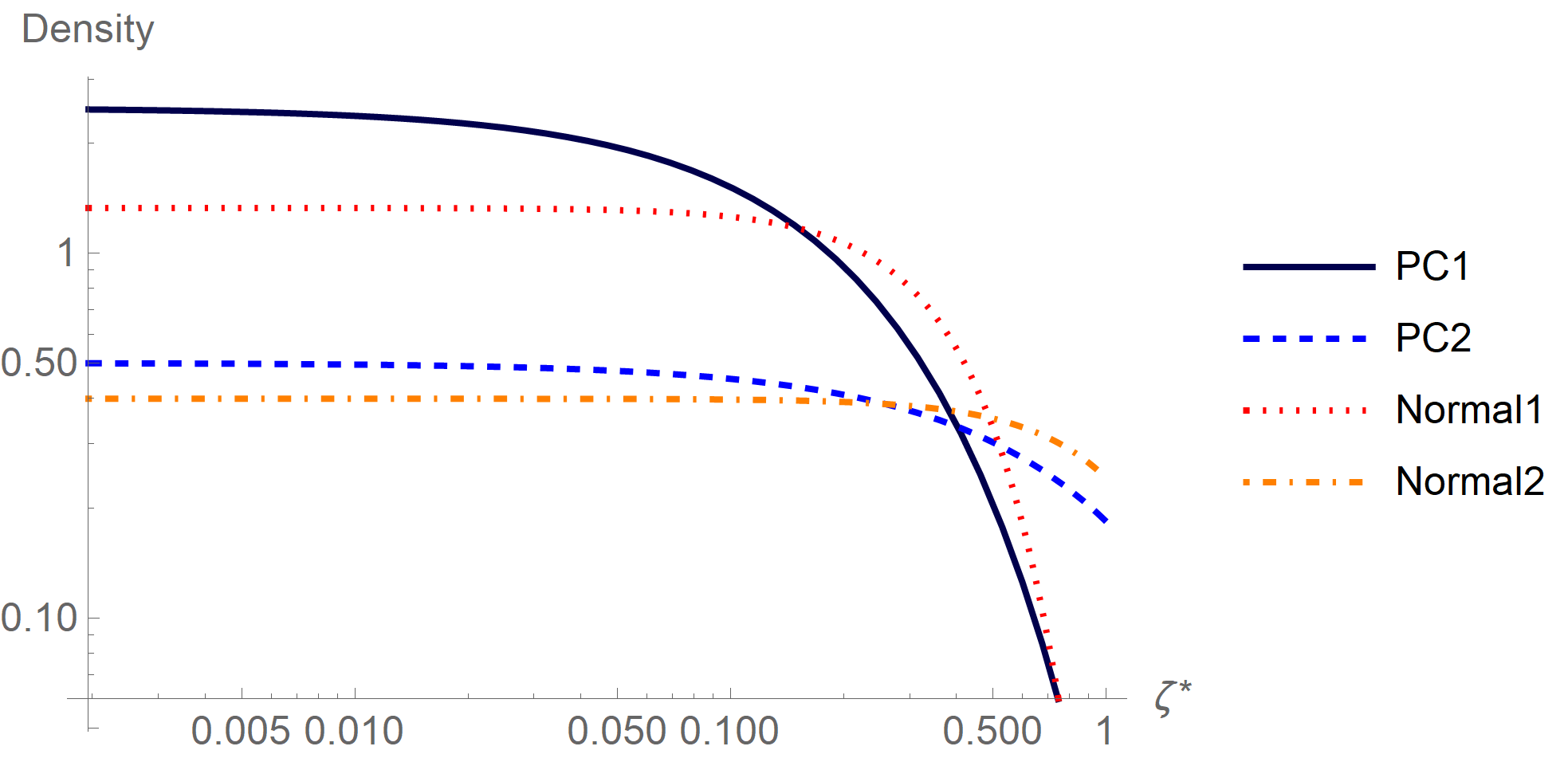}
  \caption{Prior distributions for the two flexibility parameters in log-log scale (we only plot for positive values of $\zeta^\star$).}
  \label{fig:sim5}
\end{figure}

The models were fitted in \emph{Stan} with the \emph{cmdstanr} interface \citep{cmdstanr}, and $N=200$ replications were run for each scenario, prior distribution, and sample size. Each replication consisted of 1000 warmup iterations and 1000 sampling iterations. The IG1 and IG2 priors for $\eta^\star$ were chosen so that they would have the same mean as the PC1 and PC2 priors, and therefore they could potentially achieve the same level of contraction to the Gaussian model. Likewise, the variances of the N1 and N2 priors assigned to $\zeta^\star$ are the same as the variances of the PC1 and PC2 priors for $\zeta^\star$.  These prior distributions are plotted in Fig. \ref{fig:sim5}.

 








\subsection{Results}

Fig. \ref{fig:simsce1} shows the posterior means and widths of the credible intervals for the parameters $\eta^\star$ and $\zeta^\star$.  For brevity, we only show the results of scenarios 1 and 3 and for the sample sizes of 50, 100, and 1000. The remaining figures can be found in Appendix F. For sample sizes up to 500 we see a positive bias in the estimation of $\eta^\star$ when using the uniform or Jeffreys priors, namely we have posterior means larger than 1, indicating a clear departure from Gaussianity in the latent field, where there is none. The posterior means of $\eta^\star$ are smaller when utilizing the PC priors compared to the Jeffrey priors. The reason for this is that near the base model the PC prior can be seen as a tilted Jeffrey priors \citep{simpson2017penalising}, $\pi(\eta^\star)=I(\eta^\star)^{1 / 2} \exp (-\theta_\eta m(\eta^\star))$, where $I(\eta^\star)$ is the Fisher information, and \mbox{$m(\eta^\star)=\int_{0}^{\eta^\star} \sqrt{I(s)} ds$}.

The differences in the posterior means of $\eta^\star$ between the PC priors and the Inverse Gamma priors are not substantial because the IG priors were scaled so they would have the same mean as the PC priors. Although the posterior credible intervals of $\eta^\star$  (based on the 5\% and 95\% quantiles often have smaller widths when utilizing the PC priors, compared to the IG priors, which is partly because the IG priors are more ``spread out" (see Fig. \ref{fig:sim5}). The PC and IG priors for $\eta^\star$  lead to posterior means that are closer to 0, which are more consistent between replications and with smaller widths for the credible intervals). Similar observations apply for the posterior inferences of $\zeta^\star$.

\begin{figure}[]
    \centering
 \includegraphics[width=\linewidth]{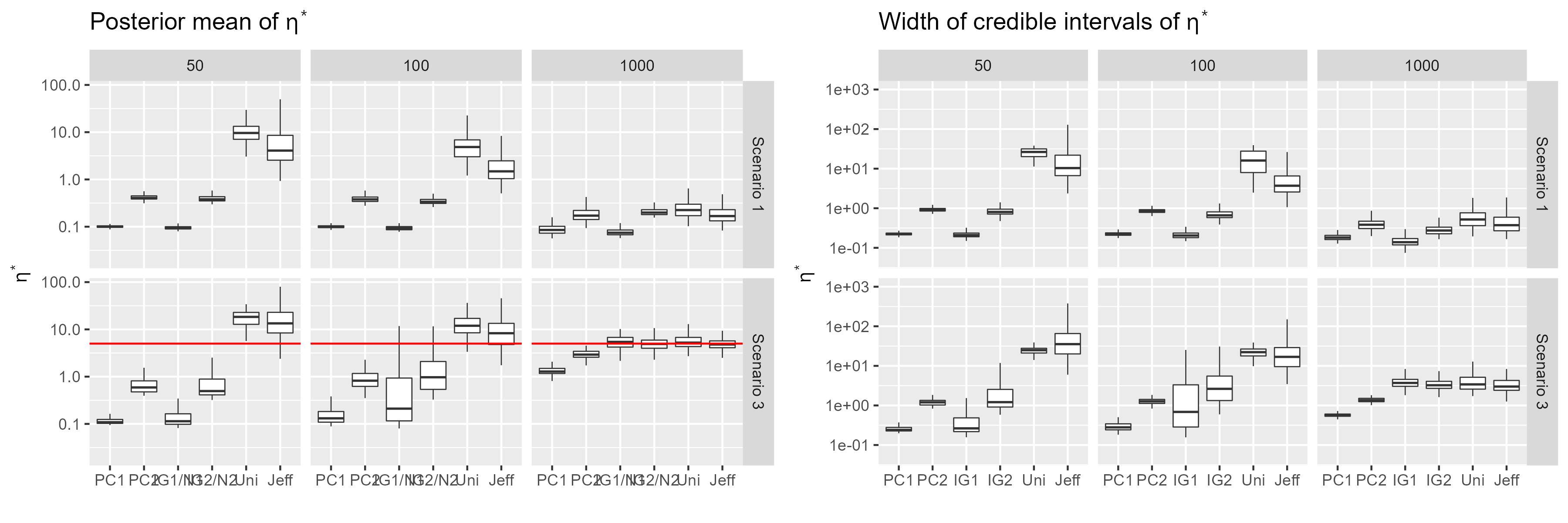}
  \includegraphics[width=\linewidth]{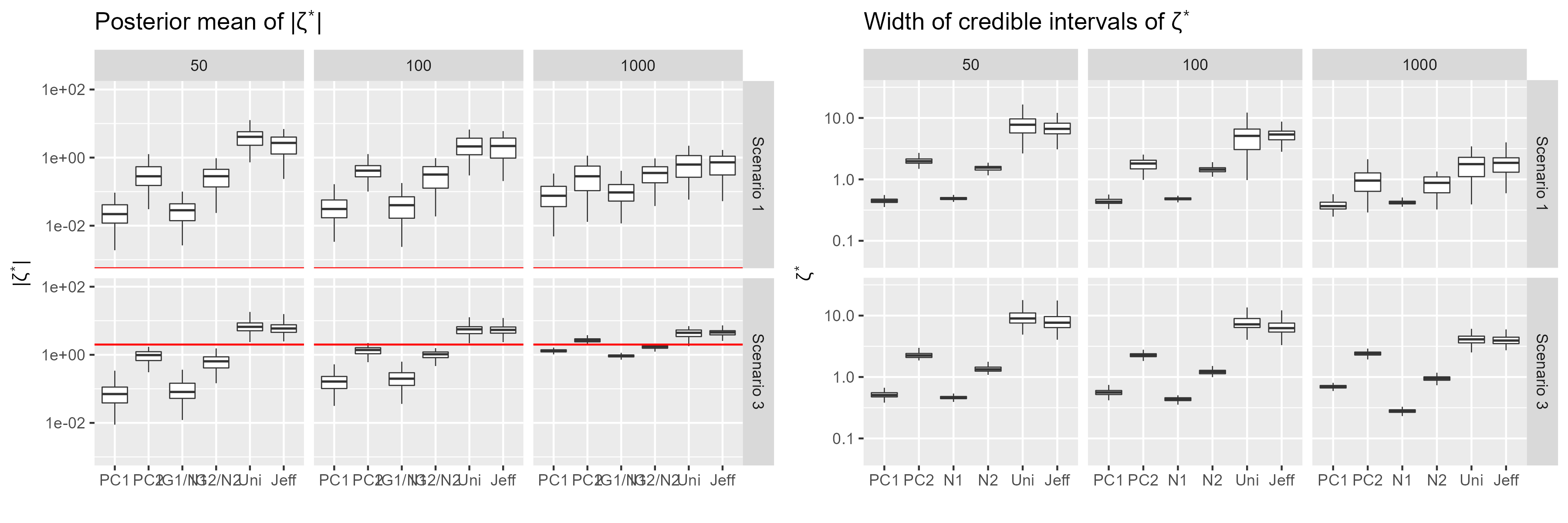}
  \caption{Histograms of the posterior means (top) and widths of the posterior credible intervals (bottom) for $\eta^\star$ (left) and $\zeta^\star$ (right) and for different sample sizes, prior configurations, and scenarios of the simulation set 1.}
  \label{fig:simsce1}
\end{figure}

\subsection{Additional simulation studies}

The results of simulation set 2 can be found in Appendix F. The PC priors led to an estimation that is less sensitive to jumps in the processes than the other priors. Overall, the two simulation studies suggest that the PC priors perform well in a variety of scenarios, leading to a more robust estimation with regards to irregularities in the data and giving preference to the Gaussian model, while at the same time allowing for non-Gaussianity if there is enough support in the data.


\section{Application} \label{sect:appli}

In this section, we illustrate the impact of the PC priors on a geostatistics application and how they achieve the sought contraction towards the Gaussian model when there is not enough convincing evidence in the data of non-Gaussianity.  

\subsection{Dataset and model implementation}

Fig. \ref{fig:app1mesh} shows temperature and pressure measurements at 157 different locations in the North American Pacific Northwest, where the sample mean was subtracted from the data in both cases. \citet{bolin2020multivariate} considered a Gaussian model for the temperature data, while the pressure data appeared to have some localized spikes and short-range variations, which were better captured with a non-Gaussian NIG model. 




\begin{figure}
    \centering
    \includegraphics[width=0.49\linewidth, height = 6cm]{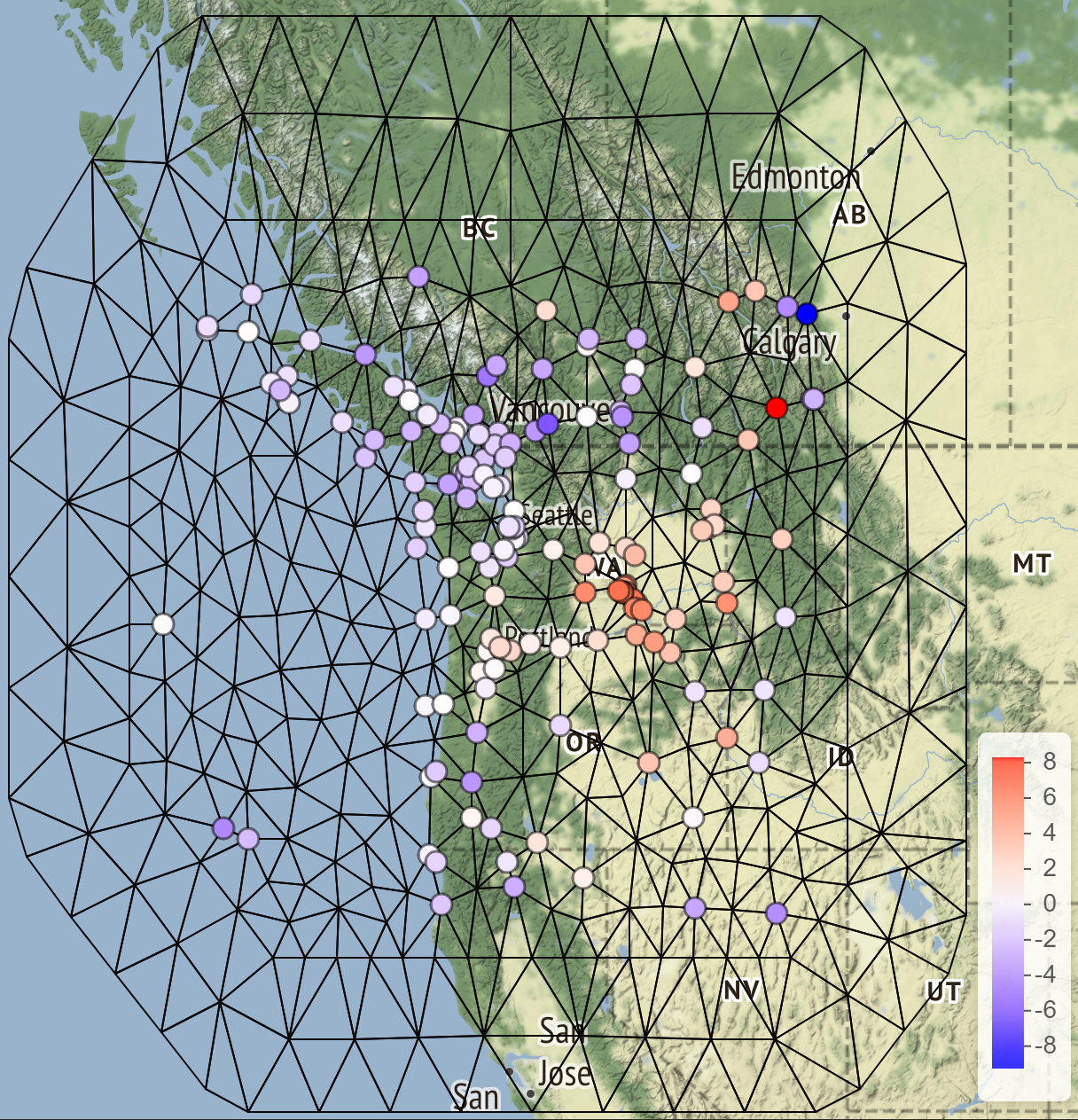}      \includegraphics[width=0.49\linewidth, height = 6cm]{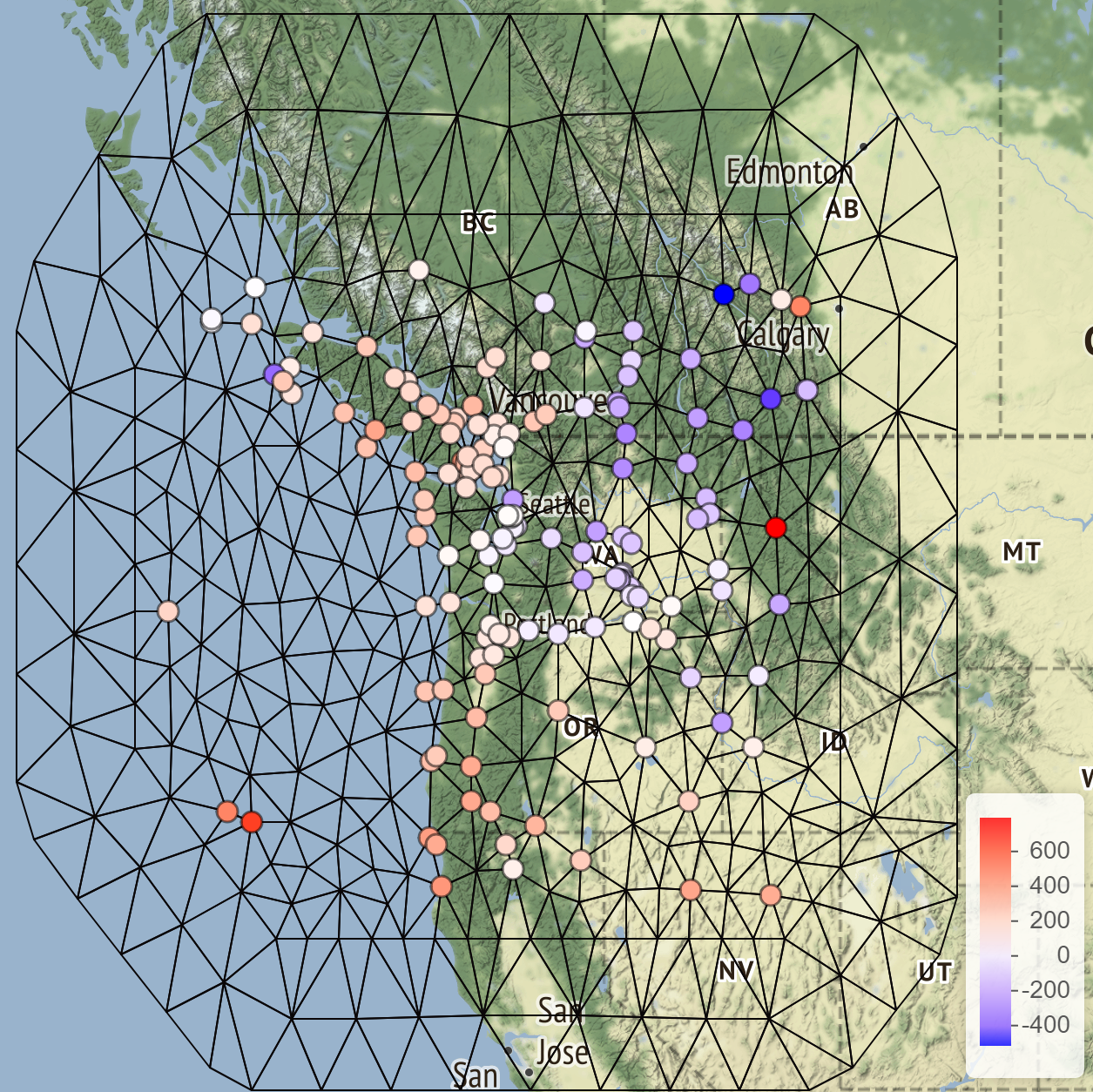}
  \caption{Measurements of temperature (left) and pressure (right) data subtracted by their sample means with the FEM mesh on the background.}
  \label{fig:app1mesh}
\end{figure}

We consider the geostatistical modelling paradigm of \citet{wallin2015geostatistical} where the field of interest $\mathbf{X}(\mathbf{s})$ is observed at $N$ locations $\mathbf{s}_1, \dotsc, \mathbf{s}_N,$ generating observations $y_1, \dotsc, y_N,$ that follow $\mathbf{y} = \mathbf{A}\mathbf{x} + \boldsymbol{\epsilon}$, where $\mathbf{A}$ is the projector matrix. The vector $\boldsymbol{\epsilon}$ is i.i.d. Gaussian noise with variance $\sigma_{\boldsymbol{\epsilon}}^2$ and $\mathbf{x}$ is a non-Gaussian random field ($\mathbf{D}\mathbf{x}=\sigma_{\mathbf{x}}\mathbf{\Lambda}$), where $\mathbf{D}$ is obtained by the FEM approximation on the Matérn SPDE with $\alpha=2$ (see subsection \ref{sect:continuousapp}). We only considered non-Gaussian models with NIG driving noise.  The mesh used for the FEM approximation is shown in Fig. \ref{fig:app1mesh}, and it was built with the $inla.mesh.2d$ function of the \emph{R-INLA} package \citep{rue2007approximate} and consisted of 394 nodes. The models were again implemented in \emph{Stan} with the \emph{cmdstanr} interface, and 4 parallel chains were run with 500 warmup iterations and 1000 sampling iterations.


The PC priors of $\eta^\star$ and $\zeta^\star$ are exponential and Laplace, respectively, and to study the impact of different calibration choices for the priors on the posterior inferences, we assigned two sets of PC priors. The first set was calibrated by choosing $\alpha_\eta= 0.01$ and $P(|\zeta^\star|>0.3 ) = 0.01$, a conservative choice which should lead to a significant contraction, where $\alpha_\eta$ is the likelihood of having twice as much large jumps in the process compared to the Gaussian process (see Appendix E). For the second set we chose $\alpha_\eta= 0.95$ and $P(|\zeta^\star|>4 ) = 0.01$, leading to near-uniform priors.  We refer to each model as  $\mathcal{M}_{PC}$ and $\mathcal{M}_{Unif}$, respectively. There were no warning messages in the $Stan$ program output, and the diagnostics indicated a good mixing, namely a split-$\widehat{R}$ smaller than 1.05 for the model parameters and large effective sample sizes. We also tried uniform priors for $\eta^\star$ and $\zeta^\star$, but the chains did not converge.



\subsection{Estimation results}

Table \ref{table:app1} shows that the posterior means of $\eta^\star$ and $\zeta^\star$ are closer to 0  for the first set of conservative PC priors, compared to the second set of near-uniform priors.  The posterior credible intervals of $\zeta^\star$ include the value 0, so a symmetric model seems adequate for both datasets. The posterior means and standard deviations of the field $X(\mathbf{s})$ are plotted in Fig. \ref{fig:posteriorplots} in a prediction grid consisting of 100000 nodes for the $\mathcal{M}_{PC}$ model. We can observe a smoother field for the temperature data and several localized spikes for the pressure data.



\begin{table}[]
\setlength\tabcolsep{1.5pt} 
\centering
\begin{tabular}{ccccc}
\cline{2-5}
\textbf{}                       & \multicolumn{2}{c}{\textbf{Prior choice 1 (conservative)}} & \multicolumn{2}{c}{\textbf{Prior choice 2 (near-uniform)}} \\ \hline
\textbf{}                       & \textbf{Temperature}       & \textbf{Pressure}             & \textbf{Temperature}       & \textbf{Pressure}             \\ \hline
$\sigma_{\boldsymbol{\epsilon}}$ & 0.79 (0.67, 0.93)          & 63.00 (54.59, 72.46)          & 0.79 (0.68, 0.94)          & 62.52 (54.74, 71.07)          \\
$\sigma_x$                      & 8.25 (6.25, 10.64)         & 336.77 (257.84, 427.72)       & 10.19 (2.09, 7.20)         & 469.34 (325.10, 649.12)       \\
$\kappa$                        & 0.93 (0.60, 1.32)          & 0.37 (0.25, 0.50)             & 1.11 (0.78, 1.51)          & 0.50 (0.35, 0.66)             \\
$\eta^\star$                    & 0.21 (0.02, 0.55)          & 2.04 (0.70, 4.05)             & 4.35 (1.02, 10.70)          & 22.01 (5.91, 48.22)           \\
$\zeta^\star$                     & 0.02 (-0.11, 0.16)         & -0.01 (-0.11, 0.10)           & 0.09 (-0.27, 0.46)         & -0.25 (-0.69, 0.19)           \\ \hline
\end{tabular}
\caption{Posterior means and $95\%$ credible intervals of the  Matérn SPDE model parameters driven by NIG noise for the two datasets and prior choices.\vspace{0.5cm}
 }
\label{table:app1}
\end{table}

\begin{figure}[]
\begin{tabular}{cc}
  \includegraphics[width=0.45\linewidth, height = 4.7cm]{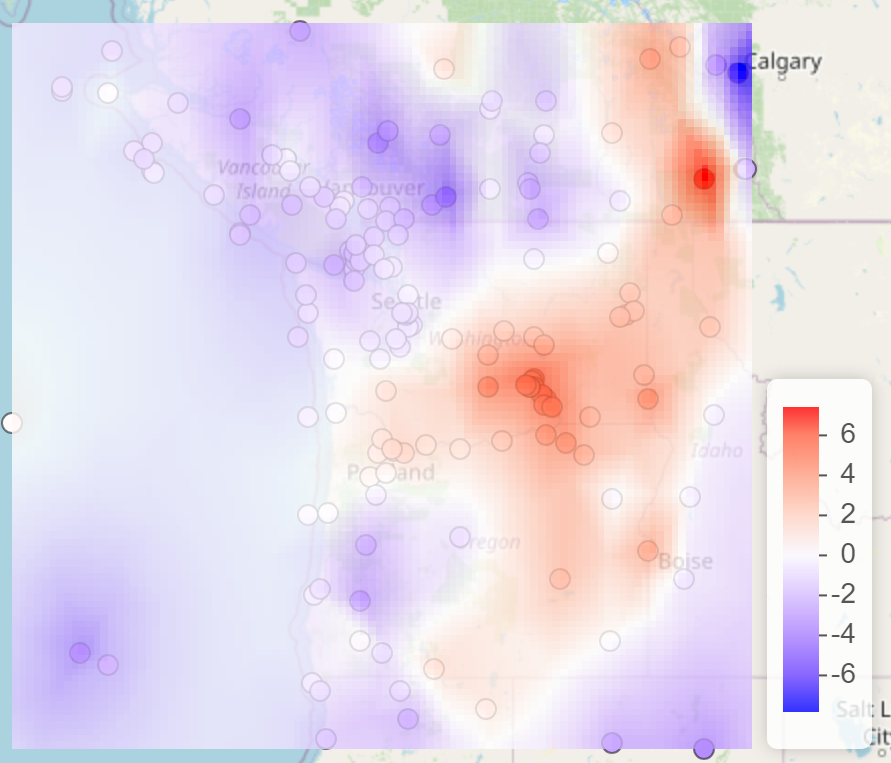} &   \includegraphics[width=0.45\linewidth, height = 4.7cm]{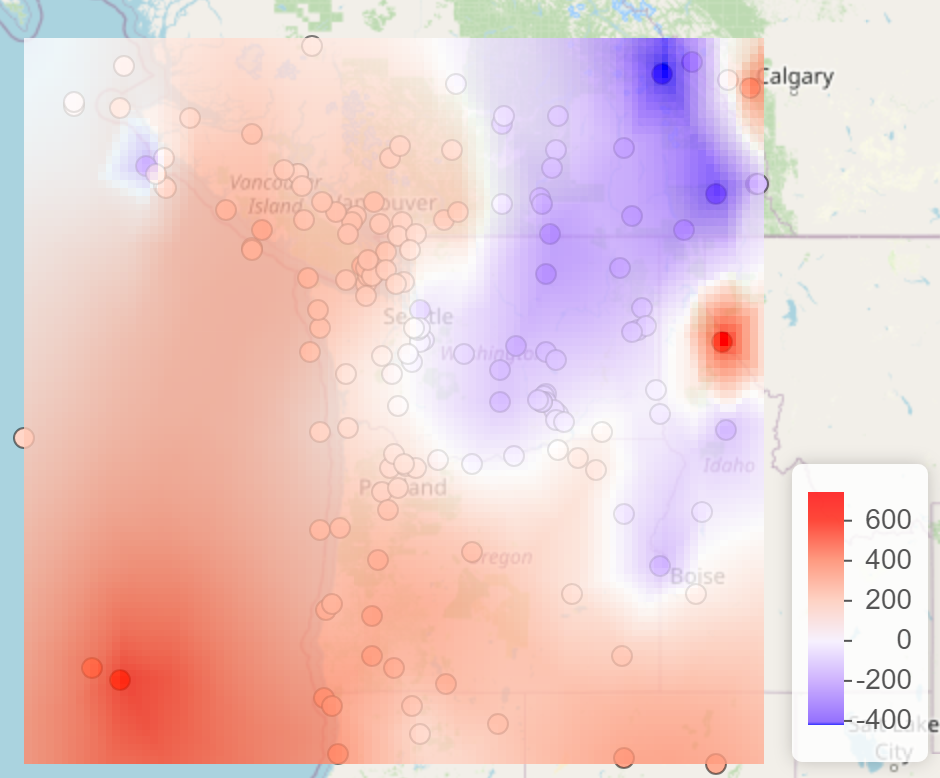} \\
(a)  & (b)  \\
 \includegraphics[width=0.45\linewidth, height = 4.7cm]{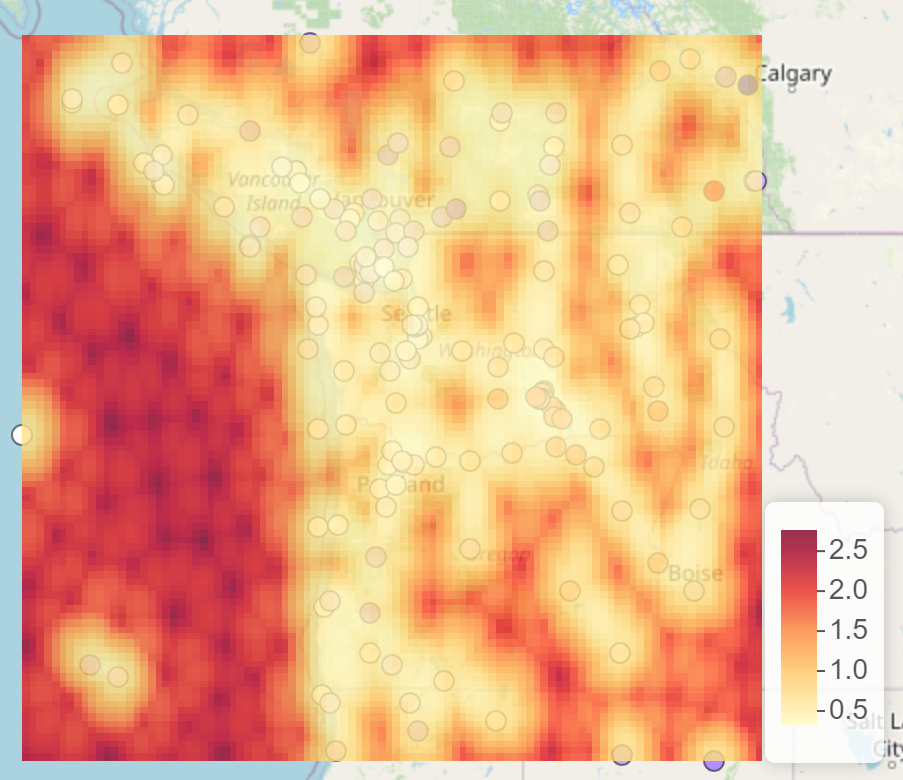} &   \includegraphics[width=0.45\linewidth, height = 4.7cm]{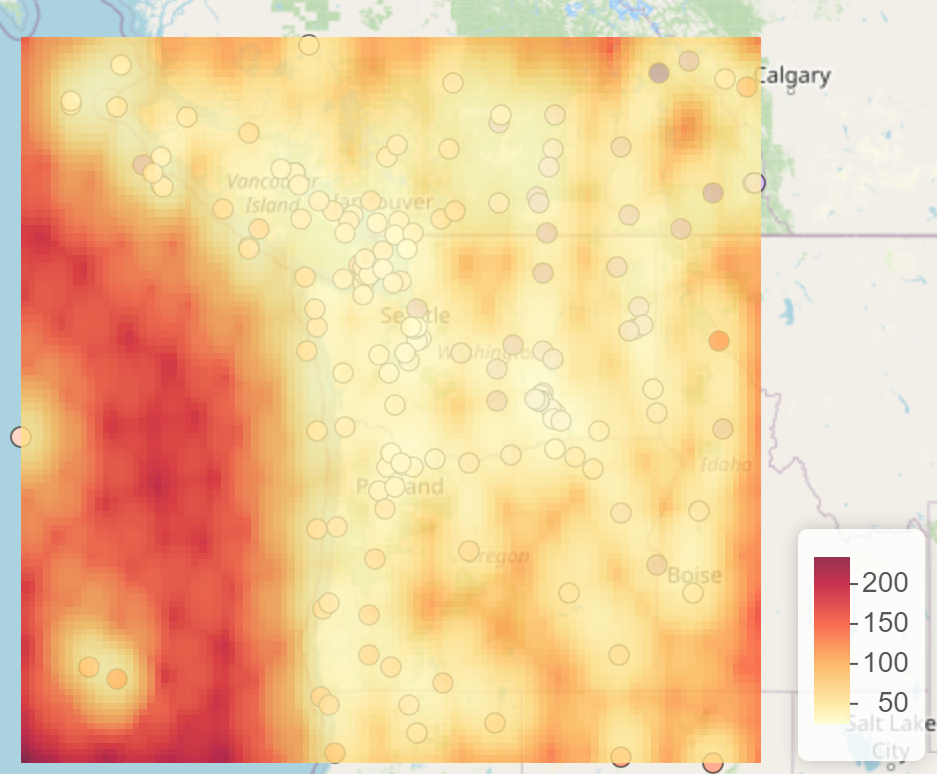} \\
(c) & (d) \\
\end{tabular}
\caption{
Plots related to temperature and pressure data are on the left and right, respectively. Plots (a)-(b) show the posterior means and plots (c)-(d) show the standard deviations of the latent field $X(\mathbf{s})$ for the first set of conservative PC priors.}
\label{fig:posteriorplots}
\end{figure}

To compare the performance of the models we performed a leave-one-out cross-validation (loocv) study. We also compared the two non-Gaussian models we fitted before ($\mathcal{M}_{PC}$ and $\mathcal{M}_{Unif}$) with a Gaussian model for the latent vector $\mathbf{x}$ ($\mathcal{M}_{Gauss}$). Performing a loovc study for non-Gaussian models by refitting a $Stan$ model at every held-out observation can be very expensive. The computation of the loocv estimates is made efficient in $Stan$ through the  $loo$ package \citep{vehtari2017practical}, however, the $loo$ function released several warnings that the approximated loo estimates were not reliable. Instead, we performed a pseudo loocv study similar to the ones performed in  \cite{bolin2020multivariate} and  \cite{bolin2020rational}, which assumes that the posterior distributions of the parameters when removing observation $i$, $\pi(\boldsymbol{\theta}|\mathbf{y}_{-i})$, is equal in distribution to $\pi(\boldsymbol{\theta}|\mathbf{y})$, where $\boldsymbol{\theta} = (\sigma_x,\sigma_y,\kappa,\eta^\star,\zeta^\star)$. Still, the posterior distribution of the latent field $\mathbf{x}$ and of the mixing variables $\mathbf{V}$ is affected when removing an observation from the dataset. 

\begin{table}[]
\centering
\begin{tabular}{ccccc}
\textbf{Data}                & \textbf{Model}          & \textbf{MSE} & \textbf{MAE} & \textbf{CRPS} \\ \hline
\multirow{3}{*}{Temperature} & Gaussian                & 4.477        & 1.506        & 0.859         \\
                             & NIG with PC prior       & 4.331        & 1.471        & 0.856         \\
                             & NIG with Unif. prior    & 4.881        & 1.491        & 0.863         \\ \hline
\multirow{3}{*}{Pressure}    & Gaussian model          & 24036.848    & 98.647       & 54.893        \\
                             & NIG model with PC prior & 20246.371     & 87.207       & 53.339        \\
                             & NIG model Unif. prior   & 21155.342     & 87.743       & 52.887   \\ \hline    
\end{tabular}
\caption{Mean squared error (MSE), mean absolute error (MAE), and continuously ranked probability score (CRPS) for the leave-one-out cross-validation predictions. }
\label{tab:loocv}
\end{table}

We computed the mean squared error (MSE), mean absolute error (MAE), and continuously ranked probability score (CRPS) of \cite{gneiting2007strictly} between the observed data $y_i$ and the held-out predictions for each observation $i=1,\dotsc,157$, and then averaged the results.  We show these results in Table \ref{tab:loocv}. The NIG model with near-uniform priors had a higher MSE and CRPS than the Gaussian model for the temperature data, so a very flexible non-Gaussian model for the temperature data did not translate into higher predictive performance.  If we look at the MSE and MAE estimates, the NIG model with PC priors had the best predictive power for both datasets, while the CRPS estimates indicate that the NIG model with near-uniform priors performed the best for the pressure data. These results suggest that the NIG model with PC priors did not overfit the data and had the highest predictive power if we use the posterior mean of the latent field $\mathbf{X}(\mathbf{s})$ as a spatial predictor.

\section{Conclusions and discussion} \label{sect:discussion}

There is a need for an inferential framework that shrinks models driven by non-Gaussian noise towards Gaussianity in order to avoid overfitting and considering a non-Gaussian model when there is not enough evidence in the data for asymmetry and leptokurtosis. We have proposed achieving this contraction by selecting weakly informative priors that penalize deviations from the simpler Gaussian model, quantified by the KLD. With an appropriate parameterization of the non-Gaussian models, this approach leads to priors that behave essentially like a LASSO penalty on the non-Gaussian flexibility parameters. 

The processes presented in this paper are more flexible, while at the same time being tractable, and are an attractive inherently robust alternative to Gaussian processes when the adequacy of the Gaussianity assumption is questionable.
\cite{barndorff1981hyperbolic} suggested that the GH distribution is well-qualified for robustness studies when the deviation from Gaussianity is in the form of asymmetry and leptokurtosis, and work is currently being done to investigate the robustness properties of this more general class of models.


 Future work also includes approximating these non-Gaussian models in the \emph{INLA} framework \citep{rue2007approximate}. A direct implementation of these models in \emph{INLA} is not possible because \emph{INLA} is only feasible for latent Gaussian models with up to 20 hyperparameters. If we condition the non-Gaussian vector $\mathbf{x}$ on the mixing variables $V_i$, we obtain a latent Gaussian model but the number of the mixing variables $V_i$, which would be hyperparameters in $INLA$, is larger than 20 in almost all applications.

\vfill

\bibliographystyle{ba}
\bibliography{bibliography.bib}

\begin{thebibliography}{41}
\newcommand{\enquote}[1]{``#1''}
\expandafter\ifx\csname natexlab\endcsname\relax\def\natexlab#1{#1}\fi
\expandafter\ifx\csname url\endcsname\relax
  \def\url#1{{\tt #1}}\fi
\expandafter\ifx\csname urlprefix\endcsname\relax\def\urlprefix{URL }\fi
\ifx\endbibitem\undefined \let\endbibitem\relax\fi

\bibitem[{{\AA}berg and Podg{\'o}rski(2011)}]{aaberg2011class}
{\AA}berg, S. and Podg{\'o}rski, K. (2011).
\newblock \enquote{{A class of non-Gaussian second order random fields}.}
\newblock {\em Extremes\/}, 14(2): 187--222.
\endbibitem

\bibitem[{Asar et~al.(2020)Asar, Bolin, Diggle, and Wallin}]{asar2020linear}
Asar, {\"O}., Bolin, D., Diggle, P.~J., and Wallin, J. (2020).
\newblock \enquote{{Linear mixed effects models for non-Gaussian continuous
  repeated measurement data}.}
\newblock {\em Journal of the Royal Statistical Society: Series C (Applied
  Statistics)\/}, 69(5): 1015--1065.
\endbibitem

\bibitem[{Bakka et~al.(2018)Bakka, Rue, Fuglstad, Riebler, Bolin, Illian,
  Krainski, Simpson, and Lindgren}]{bakka2018spatial}
Bakka, H., Rue, H., Fuglstad, G.-A., Riebler, A., Bolin, D., Illian, J.,
  Krainski, E., Simpson, D., and Lindgren, F. (2018).
\newblock \enquote{{Spatial modeling with R-INLA: A review}.}
\newblock {\em Wiley Interdisciplinary Reviews: Computational Statistics\/},
  10(6): e1443.
\endbibitem

\bibitem[{Barndorff-Nielsen(1978)}]{barndorff1978hyperbolic}
Barndorff-Nielsen, O. (1978).
\newblock \enquote{{Hyperbolic distributions and distributions on hyperbolae}.}
\newblock {\em Scandinavian Journal of statistics\/}, 151--157.
\endbibitem

\bibitem[{Barndorff-Nielsen and Blaesild(1981)}]{barndorff1981hyperbolic}
Barndorff-Nielsen, O. and Blaesild, P. (1981).
\newblock \enquote{Hyperbolic distributions and ramifications: Contributions to
  theory and application.}
\newblock In {\em Statistical distributions in scientific work\/}, 19--44.
  Springer.
\endbibitem

\bibitem[{Barndorff-Nielsen(1997)}]{barnig}
Barndorff-Nielsen, O.~E. (1997).
\newblock \enquote{{Normal Inverse Gaussian Distributions and Stochastic
  Volatility Modelling}.}
\newblock {\em Scandinavian Journal of Statistics\/}, 24(1): 1--13.
\newline\urlprefix\url{http://www.jstor.org/stable/4616433}
\endbibitem

\bibitem[{Barndorff-Nielsen(2001)}]{Barndorffsuper}
--- (2001).
\newblock \enquote{{Superposition of Ornstein--Uhlenbeck Type Processes}.}
\newblock {\em Theory of Probability \& Its Applications\/}, 45(2): 175--194.
\endbibitem

\bibitem[{Barndorff-Nielsen et~al.(2012)Barndorff-Nielsen, Mikosch, and
  Resnick}]{barndorff2012levy}
Barndorff-Nielsen, O.~E., Mikosch, T., and Resnick, S.~I. (2012).
\newblock {\em {L{\'e}vy processes: theory and applications}\/}.
\newblock Springer Science \& Business Media.
\endbibitem

\bibitem[{Barndorff-Nielsen and Shephard(2001)}]{nongaussianbarnsdorf}
Barndorff-Nielsen, O.~E. and Shephard, N. (2001).
\newblock \enquote{{Non-Gaussian Ornstein-Uhlenbeck-Based Models and Some of
  Their Uses in Financial Economics}.}
\newblock {\em Journal of the Royal Statistical Society. Series B (Statistical
  Methodology)\/}, 63(2): 167--241.
\newline\urlprefix\url{http://www.jstor.org/stable/2680596}
\endbibitem

\bibitem[{Barndorff-Nielsen and Shephard(2002)}]{bareco}
--- (2002).
\newblock \enquote{{Econometric Analysis of Realized Volatility and Its Use in
  Estimating Stochastic Volatility Models}.}
\newblock {\em Journal of the Royal Statistical Society. Series B (Statistical
  Methodology)\/}, 64(2): 253--280.
\newline\urlprefix\url{http://www.jstor.org/stable/3088799}
\endbibitem

\bibitem[{Bibby and Sørensen(2003)}]{BIBBY2003211}
Bibby, B.~M. and Sørensen, M. (2003).
\newblock \enquote{{Chapter 6 - Hyperbolic Processes in Finance}.}
\newblock In Rachev, S.~T. (ed.), {\em Handbook of Heavy Tailed Distributions
  in Finance\/}, volume~1 of {\em Handbooks in Finance\/}, 211--248. Amsterdam:
  North-Holland.
\newline\urlprefix\url{https://www.sciencedirect.com/science/article/pii/B978044450896650008X}
\endbibitem

\bibitem[{Bolin(2014)}]{bolin2014spatial}
Bolin, D. (2014).
\newblock \enquote{{Spatial Mat{\'e}rn fields driven by non-Gaussian noise}.}
\newblock {\em Scandinavian Journal of Statistics\/}, 41(3): 557--579.
\endbibitem

\bibitem[{Bolin and Kirchner(2020)}]{bolin2020rational}
Bolin, D. and Kirchner, K. (2020).
\newblock \enquote{The rational SPDE approach for Gaussian random fields with
  general smoothness.}
\newblock {\em Journal of Computational and Graphical Statistics\/}, 29(2):
  274--285.
\endbibitem

\bibitem[{Bolin and Wallin(2020)}]{bolin2020multivariate}
Bolin, D. and Wallin, J. (2020).
\newblock \enquote{{Multivariate type G Mat{\'e}rn stochastic partial
  differential equation random fields}.}
\newblock {\em Journal of the Royal Statistical Society: Series B (Statistical
  Methodology)\/}, 82(1): 215--239.
\endbibitem

\bibitem[{Cover and Thomas(2006)}]{informationtheory}
Cover, T.~M. and Thomas, J.~A. (2006).
\newblock {\em {Elements of Information Theory (Wiley Series in
  Telecommunications and Signal Processing)}\/}.
\newblock USA: Wiley-Interscience.
\endbibitem

\bibitem[{Deschamps(2012)}]{GHSSTGH}
Deschamps, P.~J. (2012).
\newblock \enquote{{Bayesian Estimation of Generalized Hyperbolic Skewed
  Student GARCH Models}.}
\newblock {\em Comput. Stat. Data Anal.\/}, 56(11): 3035–3054.
\newline\urlprefix\url{https://doi.org/10.1016/j.csda.2011.10.021}
\endbibitem

\bibitem[{Dhull and Kumar(2021)}]{NIGAR}
Dhull, M.~S. and Kumar, A. (2021).
\newblock \enquote{{Normal inverse Gaussian autoregressive model using EM
  algorithm}.}
\newblock {\em International Journal of Advances in Engineering Sciences and
  Applied Mathematics\/}.
\endbibitem

\bibitem[{Diggle et~al.(2014)Diggle, Sousa, and Asar}]{diggle2014}
Diggle, P.~J., Sousa, I., and Asar, {\"O}. (2014).
\newblock \enquote{{Real-time monitoring of progression towards renal failure
  in primary care patients}.}
\newblock {\em Biostatistics\/}, 16(3): 522--536.
\newline\urlprefix\url{https://doi.org/10.1093/biostatistics/kxu053}
\endbibitem

\bibitem[{Gabry and Ce\v{s}novar(2021)}]{cmdstanr}
Gabry, J. and Ce\v{s}novar, R. (2021).
\newblock {\em {cmdstanr: R Interface to CmdStan}\/}.
\newline\urlprefix\url{https://mc-stan.org/cmdstanr}
\endbibitem

\bibitem[{Gelman et~al.(2017)Gelman, Simpson, and Betancourt}]{gelman2017prior}
Gelman, A., Simpson, D., and Betancourt, M. (2017).
\newblock \enquote{The prior can often only be understood in the context of the
  likelihood.}
\newblock {\em Entropy\/}, 19(10): 555.
\endbibitem

\bibitem[{Ghasami et~al.(2020)Ghasami, Khodadadi, and Maleki}]{GHAR}
Ghasami, S., Khodadadi, Z., and Maleki, M. (2020).
\newblock \enquote{{Autoregressive processes with generalized hyperbolic
  innovations}.}
\newblock {\em Communications in Statistics - Simulation and Computation\/},
  49(12): 3080--3092.
\endbibitem

\bibitem[{Gneiting and Raftery(2007)}]{gneiting2007strictly}
Gneiting, T. and Raftery, A.~E. (2007).
\newblock \enquote{Strictly proper scoring rules, prediction, and estimation.}
\newblock {\em Journal of the American statistical Association\/}, 102(477):
  359--378.
\endbibitem

\bibitem[{Hammerstein(2016)}]{hammerstein2016tail}
Hammerstein, E. A.~v. (2016).
\newblock \enquote{Tail behaviour and tail dependence of generalized hyperbolic
  distributions.}
\newblock In {\em Advanced modelling in mathematical finance\/}, 3--40.
  Springer.
\endbibitem

\bibitem[{Karlsson et~al.(2021)Karlsson, Mazur, and
  Nguyen}]{karlsson2021vector}
Karlsson, S., Mazur, S., and Nguyen, H. (2021).
\newblock \enquote{{Vector autoregression models with skewness and heavy
  tails}.}
\newblock {\em arXiv preprint arXiv:2105.11182\/}.
\endbibitem

\bibitem[{Ken-Iti(1999)}]{ken1999levy}
Ken-Iti, S. (1999).
\newblock {\em {L{\'e}vy processes and infinitely divisible distributions}\/}.
\newblock Cambridge university press.
\endbibitem

\bibitem[{Lindgren et~al.(2011)Lindgren, Rue, and
  Lindstr{\"o}m}]{lindgren2011explicit}
Lindgren, F., Rue, H., and Lindstr{\"o}m, J. (2011).
\newblock \enquote{{An explicit link between Gaussian fields and Gaussian
  Markov random fields: the stochastic partial differential equation
  approach}.}
\newblock {\em Journal of the Royal Statistical Society: Series B (Statistical
  Methodology)\/}, 73(4): 423--498.
\endbibitem

\bibitem[{Mat\'{e}rn(1960)}]{MR0169346}
Mat\'{e}rn, B. (1960).
\newblock {\em Spatial variation: {S}tochastic models and their application to
  some problems in forest surveys and other sampling investigations\/}.
\newblock Statens Skogsforskningsinstitut, Stockholm.
\newblock Meddelanden Fran Statens Skogsforskningsinstitut, Band 49, Nr. 5.
\endbibitem

\bibitem[{Nakajima and Omori(2012)}]{sv2}
Nakajima, J. and Omori, Y. (2012).
\newblock \enquote{{Stochastic volatility model with leverage and
  asymmetrically heavy-tailed error using GH skew Student’s t-distribution}.}
\newblock {\em Computational Statistics \& Data Analysis\/}, 56(11):
  3690--3704.
\newline\urlprefix\url{https://www.sciencedirect.com/science/article/pii/S0167947310002859}
\endbibitem

\bibitem[{Niekerk and Rue(2021)}]{niekerkskewprobit}
Niekerk, J. and Rue, H. (2021).
\newblock \enquote{{Skewed Probit Regression — Identifiability, Contraction,
  and Reformulation}, volume = {19}.}
\newblock {\em REVSTAT\/}, (1): 1–--22.
\endbibitem

\bibitem[{Paolella(2007)}]{paolella2007intermediate}
Paolella, M.~S. (2007).
\newblock {\em Intermediate probability: A computational approach\/}.
\newblock John Wiley \& Sons.
\endbibitem

\bibitem[{Prause et~al.(1999)}]{prause1999generalized}
Prause, K. et~al. (1999).
\newblock \enquote{The generalized hyperbolic model: Estimation, financial
  derivatives, and risk measures.}
\newblock Ph.D. thesis, Citeseer.
\endbibitem

\bibitem[{Rue and Held(2005)}]{rue2005gaussian}
Rue, H. and Held, L. (2005).
\newblock {\em {Gaussian Markov random fields: theory and applications}\/}.
\newblock CRC press.
\endbibitem

\bibitem[{Rue and Martino(2007)}]{rue2007approximate}
Rue, H. and Martino, S. (2007).
\newblock \enquote{{Approximate Bayesian inference for hierarchical Gaussian
  Markov random field models}.}
\newblock {\em Journal of statistical planning and inference\/}, 137(10):
  3177--3192.
\newline\urlprefix\url{https://www.r-inla.org/}
\endbibitem

\bibitem[{Simpson et~al.(2017)Simpson, Rue, Riebler, Martins, and
  S{\o}rbye}]{simpson2017penalising}
Simpson, D., Rue, H., Riebler, A., Martins, T.~G., and S{\o}rbye, S.~H. (2017).
\newblock \enquote{{Penalising model component complexity: A principled,
  practical approach to constructing priors}.}
\newblock {\em Statistical science\/}, 1--28.
\endbibitem

\bibitem[{{Stan Development Team}(2020)}]{stan}
{Stan Development Team} (2020).
\newblock \enquote{{Stan Modeling Language Users Guide and Reference Manual,
  2.28}.}
\newline\urlprefix\url{http://mc-stan.org/}
\endbibitem

\bibitem[{Vehtari et~al.(2017)Vehtari, Gelman, and
  Gabry}]{vehtari2017practical}
Vehtari, A., Gelman, A., and Gabry, J. (2017).
\newblock \enquote{{Practical Bayesian model evaluation using leave-one-out
  cross-validation and WAIC}.}
\newblock {\em Statistics and computing\/}, 27(5): 1413--1432.
\endbibitem

\bibitem[{Walder and Hanks(2020)}]{walder2020bayesian}
Walder, A. and Hanks, E.~M. (2020).
\newblock \enquote{{Bayesian analysis of spatial generalized linear mixed
  models with Laplace moving average random fields}.}
\newblock {\em Computational Statistics \& Data Analysis\/}, 144: 106861.
\endbibitem

\bibitem[{Wallin and Bolin(2015)}]{wallin2015geostatistical}
Wallin, J. and Bolin, D. (2015).
\newblock \enquote{{Geostatistical modelling using non-Gaussian Mat{\'e}rn
  fields}.}
\newblock {\em Scandinavian Journal of Statistics\/}, 42(3): 872--890.
\endbibitem

\bibitem[{Whittle(1963)}]{Whittle}
Whittle, P. (1963).
\newblock \enquote{{Stochastic processes in several dimensions}.}
\newblock {\em Bulletin of the International Statistical Institute\/}, 40(2):
  974--994.
\endbibitem

\bibitem[{Xie and Shen(2021)}]{sv1}
Xie, F.-C. and Shen, Y.-Y. (2021).
\newblock \enquote{{Bayesian estimation for stochastic volatility model with
  jumps, leverage effect and generalized hyperbolic skew Student's
  t-distribution}.}
\newblock {\em Communications in Statistics - Simulation and Computation\/},
  0(0): 1--18.
\endbibitem

\bibitem[{Zhu and Dunson(2017)}]{zhuintrw}
Zhu, B. and Dunson, D.~B. (2017).
\newblock \enquote{{Bayesian Functional Data Modeling for Heterogeneous
  Volatility}.}
\newblock {\em Bayesian Analysis\/}, 12(2): 335 -- 350.
\newline\urlprefix\url{https://doi.org/10.1214/16-BA1004}
\endbibitem

\end{thebibliography}

\begin{appendices}

\section{Characteristic function of the noise $\mathbf{\Lambda}$}\label{app:char}

The NIG and GAL distributions described in subsection 2.1 of the main paper are more easily defined through their characteristic functions (CFs) than their PDFs. The CF of the NIG distribution is
\begin{equation}\label{eq:CFNIG}
    \varphi_\Lambda(t) = \exp \left(-it  \tilde{\sigma} \zeta + \frac{1}{\eta} \left( 1-\sqrt{1- 2\eta(i t  \tilde{\sigma} \zeta -\tilde{\sigma}^2t^2/2)} \right) \right),
\end{equation}
and the CF of the GAL distribution is
\begin{equation}\label{eq:CFGAL}
\varphi_\Lambda(t) = e^{-it \tilde{\sigma} \zeta }\left(1- \eta(it \tilde{\sigma} \zeta -\tilde{\sigma}^2 t^2/2)\right)^{-\eta^{-1}},
\end{equation} 
where $\tilde{\sigma} = 1 / \sqrt{1+\eta \zeta^2}$. If $\Lambda_i$ is an increment of length $h_i$ of a Lévy process, the CFs of the noises are obtained by raising eqs. \eqref{eq:CFNIG} and \eqref{eq:CFGAL} to the power $h_i$.

\section{Properties of the field $\mathbf{x}$} \label{app:prop}

Consider the $n \times n$ non-singular matrix $\mathbf{D}$ and let $\mathbf{x}=[x_1,\dotsc,x_n]^T$ be a $n$-dimensional random vector defined via $\mathbf{D}\mathbf{x} = \mathbf{\Lambda}$. The marginals of $\mathbf{x}$ are given by $x_i= \sum_{j=1}^n D_{ij}^{-1}\Lambda_j$, and so the CF is $\varphi_{x_i}(t) = \prod_{j=1}^n \varphi_{\Lambda_j}(D_{ij}^{-1}t)$, with $\varphi_{\Lambda_j}(t)$ specified in Appendix \ref{app:char}. For NIG noise, it is given by
\begin{equation}\label{eq:CFmarginalNIG}
\varphi_{x_i}(t)=\exp\left(\sum_j h_j \left( -it \tilde{\sigma} \zeta D_{ij}^{-1} + \frac{1}{\eta} \left(1-\sqrt{1-2\eta(it\tilde{\sigma}\zeta D_{ij}^{-1}-t^2\tilde{\sigma}^2 (D_{ij}^{-1})^2/2 )} \right)\right)\right),
\end{equation}
and for GAL noise it is
\begin{equation}\label{eq:CFmarginalGAL}
\varphi_{x_i}(t)=\exp\left(\sum_j h_j \left( -it \tilde{\sigma} \zeta D_{ij}^{-1} - \frac{1}{\eta}\log\left(1- \eta(it \tilde{\sigma} \zeta D_{ij}^{-1} -t^2\tilde{\sigma}^2(D_{ij}^{-1})^2/2)\right)\right)\right).
\end{equation}

The moments of the marginals $x_i$ can be obtained from the previous CFs. The marginal mean is 0, and variance is $\text{V}[x_i]  = \sum_j h_j D_{ij}^{-2}$. For NIG noise, the marginal skewness ($S$) and excess kurtosis ($EK$) are:
   $$S[x_i] = \frac{3\zeta\eta}{\sqrt{1+\zeta^2\eta}}{\frac{\sum_j h_j D_{ij}^{-3}}{\left(\sum_j h_j D_{ij}^{-2}\right)^{3/2}}},$$
   $$EK[x_i]= \frac{3\eta(1+h\zeta^2\eta)}{1+\zeta^2\eta} \frac{\sum_jh_j D_{ij}^{-4}}{\left(\sum_j h_j D_{ij}^{-2}\right)^2},$$
and for GAL noise, these are:
 $$S[x_i] =\frac{\zeta\eta(3+2\zeta^2\eta)}{\sqrt{(1+\zeta^2\eta)^{3}}}{\frac{\sum_j h_j D_{ij}^{-3}}{\left(\sum_j h_j D_{ij}^{-2}\right)^{3/2}}},$$
  $$EK[x_i]= \frac{3 \eta (1 + 4 \zeta^2 \eta + 2 \zeta^4 \eta^2)}{(1 + \zeta^2 \eta)^2} \frac{\sum_jh_j D_{ij}^{-4}}{\left(\sum_j h_j D_{ij}^{-2}\right)^2}.$$

\section{Non-Gaussian moving averages} \label{app:ngma}

Subsection 4.2 of the main paper presented several continuous stationary stochastic processes that can be expressed through the SDE $\mathcal{D}X(t) = \sigma dL(t), t \in \mathbb{R}$, and here we characterize the marginal distributions of the process $X(t)$. An alternative representation of the model is given by a process convolution:
\begin{equation}\label{eq:convL}
X(t) = \int_{-\infty}^{+\infty} G(t-u) dL(u),
\end{equation}
where $G(t,u)=G(t-u)$ is the Green function of the operator $\mathcal{D}$ and $L(t)$ is the background driving Lévy process constructed so that $L(1)$ follows either a standardized NIG or GAL distribution. The Lévy process is extended to the whole real line by taking two independent copies of it and mirroring one of them at the origin. This extended process can be used to define eq. $\eqref{eq:convL}$, and so $X(t)$ can be seen as a convolution of $G$ with the increments of the process $L(t)$. The Green function of the Ornstein-Uhlenbeck process is $G(t) = \Theta(t)e^{-\kappa t}$, where $\Theta(t)$ is the Heaviside step function, and the Green function of the Matérn model in dimension $d$, with smoothness parameter $\alpha$, and spatial range parameter $\kappa$ is
$$
G_{\alpha}(\mathbf{s}, \mathbf{t})=\frac{2^{1-\frac{\alpha-d}{2}}}{(4 \pi)^{\frac{d}{2}} \Gamma\left(\frac{\alpha}{2}\right) \kappa^{\alpha-d}}(\kappa\|\mathbf{s}-\mathbf{t}\|)^{\frac{\alpha-d}{2}} K_{\frac{\alpha-d}{2}}(\kappa\|\mathbf{s}-\mathbf{t}\|).
$$
Processes $X(t)$ defined via eq. \eqref{eq:convL} are also referred to as non-Gaussian moving averages, and according to Proposition 1 of \cite{aaberg2011class} (where a different parameterization is used), the CF of the marginals of $X(t)$ when $L(t)$ is a GAL process is
$$
\varphi_{X(t)}(u)=\exp\left(\int_{-\infty}^{+\infty}-iu\tilde{\sigma} \zeta G(x) - \frac{1}{\eta}\log\left(1- \eta(iu \tilde{\sigma} \zeta G(x) -u^2\tilde{\sigma}^2G(x)^2/2)\right) dx\right).
$$
The previous expression is similar to eq. \eqref{eq:CFmarginalGAL}, where for the continuous case, one deals with the Green function $G(t)$ instead of the matrix $\mathbf{D}^{-1}$. The expressions for the marginal moments are also similar to those in Appendix \ref{app:prop}, where $\sum_j h_j D_{ij}^{-2}$ should be replaced with $\int G(x)^2 dx$, $\sum_j h_j D_{ij}^{-3}$ should be replaced with $\int G(x)^3 dx$ and so on. When the background driving Lévy process is a NIG process, the marginal CFs can be shown to be
$$
\varphi_{X(t)}(u)=\exp\left(\int_{-\infty}^{\infty} -iu \tilde{\sigma} \zeta G(x) + \frac{1}{\eta} \left(1-\sqrt{1-2\eta(iu\tilde{\sigma}\zeta G(x)-u^2\sigma^2 G(x)^2/2 )}  \right)dx\right),
$$
which follows from Proposition 2.1 of \cite{Barndorffsuper}. 

Table \ref{tab:cfs} shows the CFs and moments of the marginal distributions of several models. The marginal CF for stationary and isotropic random fields defined via the SPDE  $\mathcal{D}X(\mathbf{s})=\mathcal{{L}}(\mathbf{s}), \mathbf{s} \in \mathbb{R}^d$ driven by GAL white noise is
$$
\varphi_{X(\mathbf{s})}(u) = \exp\left(2\pi\int_{0}^{+\infty}\left(-iu\tilde{\sigma} \zeta G(r) - \frac{1}{\eta}\log(1- \eta(iu \tilde{\sigma} \zeta G(r) -u^2\tilde{\sigma}^2G(r)^22))\right) r dr\right),$$
where $G(\mathbf{s},\mathbf{t})$ is the Green function associated with $\mathcal{D}$, which for isotropic random fields depends only on $r = ||\mathbf{s} - \mathbf{t} ||$. The obvious transformations apply when the driving noise is NIG white noise. 









\begin{sidewaystable}

\def\arraystretch{1.5}
\begin{tabular}{ccccc}

\hline
\textbf{Model}               & \textbf{log} $\boldsymbol{\varphi}\boldsymbol{(u)}$  \textbf{ in symmetric case}                                                                                                                                                                                         & \textbf{Variance}             & \textbf{Skewness}                  & \textbf{Excess Kurt.}                                  \\ \hline
GAL OU $d=1$                 & $\frac{\text{Li}_2\left(-\frac{1}{2} \eta  u^2\right)}{2 \kappa \eta}$                                                                                                                                                         & $\frac{1}{2\kappa}$   & $S_{GAL}\frac{2\sqrt{2\kappa}}{3}$ & $K_{GAL}\kappa$                                    \\
GAL Matérn $d=1$, $\alpha=2$ & $\frac{\text{Li}_2\left(-\frac{\eta  u^2}{8 \kappa^2}\right)}{\kappa \eta}$                                                                                                                                                    & $\frac{1}{4\kappa^3}$  & $S_{GAL}\frac{2\sqrt{\kappa}}{3}$   & $K_{GAL}\frac{\kappa}{2}$                           \\
GAL Matérn $d=2$, $\alpha=2$   & -                                                                                                                                                                                                                                    & $ \frac{1}{4\kappa^2\pi}$ & $0.661204 S_{GAL} $                          & $K_{GAL}\frac{7 \kappa^2 \tilde{\zeta} (3)}{4 \pi }$ \\ \hline
NIG OU $d=1$                 & $\frac{-\sqrt{\eta u^2+1}+\text{csch}^{-1}\left(\sqrt{\eta} u\right)+\frac{\log (\eta)}{2}+\log \left(\frac{1}{2}\right)+\log (u)+1}{\kappa \eta}$                                                 & $ \frac{1}{2\kappa}$                    & $S_{NIG}\frac{2\sqrt{2\kappa}}{3}$                          & $K_{NIG}\kappa$                                          \\
NIG Matérn $d=1$, $\alpha=2$   & $\frac{\kappa \left(2 \sinh ^{-1}\left(\frac{2 \kappa}{\sqrt{\eta} u}\right)-2 \log (\kappa)+\log (\eta)+2 \left(\log \left(\frac{1}{4}\right)+\log (u)+1\right)\right)-\sqrt{4 \kappa^2+\eta  u^2}}{\kappa^2 \eta}$ & $\frac{1}{4\kappa^3}$                     & $S_{NIG}\frac{2\sqrt{\kappa}}{3}$                          & $K_{NIG}\frac{\kappa}{2}$                                          \\
NIG Matérn $d=2$, $\alpha=2$   & -                                                                                                                                                                                                                                    & $ \frac{1}{4\kappa^2\pi}$                     & $0.661204S_{NIG}$                          & $K_{NIG}\frac{7 \kappa^2 \tilde{\zeta}(3)}{4 \pi }$                                          \\ \hline
\end{tabular}
\caption{Logarithm of the characteristic function in the symmetric case, variance, skewness and excess kurtosis of the marginal distribution of $X$ for several stationary models. $Li(t)$ is the PolyLog function, $csch$ is the hyperbolic cosecant function, $\log$ is the complex logarithmic function, and $\tilde{\zeta}$ is the Riemann zeta function. Furthermore,
$S_{NIG} = \frac{3\zeta\eta}{\sqrt{1+\zeta^2\eta}}, \ $             
$K_{NIG} = \frac{3\eta(1+5\zeta^2\eta)}{(1+\zeta^2\eta)}, \ $  
$S_{GAL} =\frac{\zeta\eta(3+2\zeta^2\eta)}{\sqrt{(1+\zeta^2\eta)^{3}}}$, and 
$K_{GAL} = \frac{3 \eta (1 + 4 \zeta^2 \eta + 2 \zeta^4 \eta^2)}{(1 + \zeta^2 \eta)^2}.$}
\label{tab:cfs}
\end{sidewaystable}

\section{Derivation of the PC priors}

 Here we prove the two theorems presented on the paper.
 
\subsection{Proof of Theorem 3.1.} \label{app:theo1}

\begin{proof}

The KLD between the two random noise variables $\Lambda_i$ and $Z_i$ is

\begin{equation} \label{eq:KLDexp}
KLD(\Lambda_i \ || \ Z_i) = \int \pi_{\Lambda_i}(x|\eta)\log\left({\frac{\pi_{\Lambda_i}(x|\eta)}{\pi_{Z_i}(x)}}\right) dx.
\end{equation}
We start by assuming that $\Lambda_i$ is symmetric NIG noise. The Taylor expansion of the NIG density $\pi^{\text{NIG}}_{\Lambda_i}(x|\eta, \zeta=0)$ near $\eta=0$ yields the Gaussian density $\pi_{Z_i}(x)$ multiplied by a polynomial of $\eta$: 
\begin{flalign}
    \pi^{\text{NIG}}_{\Lambda_i}(x) =  \frac{h_i e^{\frac{h_i}{\eta }} K_1\left(\sqrt{\frac{1}{\eta }} \sqrt{x^2+\frac{h_i^2}{\eta }}\right)}{\pi  \eta  \sqrt{\frac{h_i^2}{\eta }+x^2}} = \pi_{Z_i}(x)\left(1 + f_1(x,\eta)\eta  + f_2(x,\eta)\eta^2 + \mathcal{O}(\eta^3)\right), \nonumber \\
     f_1(x,\eta) = \frac{(3 h_i^2-6 h_i x^2+x^4)}{8 h_i^3}, \ \ 
     f_2(x,\eta) = \frac{ \left(90 h_i^2 x^4-60 h^3 x^2-15 h_i^4-20 h_i x^6+x^8\right)}{128 h_i^{6}}. \nonumber 
\end{flalign}
Replacing the previous expansion in eq. \eqref{eq:KLDexp} yields
\begin{align*} 
KLD(\Lambda^{\text{NIG}}_i \ || \ Z_i) = \frac{3}{16h_i^2}\eta^2 -\frac{9}{16h_i^3}\eta^3+\frac{261}{128h_i^4}\eta^4 + \mathcal{O}(\eta^5).
\end{align*}

When $\Lambda_i$ is symmetric GAL noise, it is easier to start by doing the Taylor expansion of the characteristic function around $\eta=0$:
\begin{align*}
\varphi^{\text{GAL}}_{\Lambda_i}(t) &=  2^{\frac{h_i}{\eta}}\left(\frac{1}{2+\eta t^2}\right)^{\frac{h_i}{\eta}} \\
&= e^{-\frac{h_i t^2}{2}}\left( 1 + \frac{h_i t^4}{8} \eta + \frac{-16 h_i t^6 + 3 h_i^2 t^8}{384}\eta^2 + \mathcal{O}(\eta^3)\ \right).
\end{align*}
The inverse Fourier transform of the previous expansion also yields the Gaussian density $\pi_{Z_i}(x)$ multiplied by a polynomial of $\eta$:
\begin{flalign}
\pi^{\text{GAL}}_{\Lambda_i}(x) = \pi_{Z_i}(x)\left(1 + f_1(x,\eta)\eta  + f_3(x,\eta)\eta^2  +\mathcal{O}(\eta^3) \right),   \nonumber \\
f_3(x,\eta) = \frac{75 h_i^4 - 540 h_i^3x^2 + 390 h_i^2 x^4 -68 h_i x^6+3x^8}{384 h_i^{6}}.\nonumber
\end{flalign}
By utilizing the previous expansion in eq. \eqref{eq:KLDexp} one gets
\begin{align*} 
KLD(\Lambda_i^{GAL} \ || \ Z_i) = \frac{3}{16h_i^2}\eta^2 -\frac{9}{16h_i^3}\eta^3+\frac{401}{128h_i^4}\eta^4 + \mathcal{O}(\eta^5).
\end{align*}

Consequently, using the Remark 1 of the main paper (section 3), the KLD between the non-Gaussian $\mathbf{x} = \mathbf{D}^{-1} \mathbf{\Lambda}$ and Gaussian $\mathbf{x}^G = \mathbf{D}^{-1} \mathbf{Z}$ random vectors is
\begin{equation*}
KLD(\mathbf{x} \ || \ \mathbf{x}^{G}) = \sum_{i=1}^{n} KLD(\Lambda_i \ || \ Z_i) =   \frac{3}{16}\left(\sum_{i=1}^n\frac{1}{h_i^2}\right)\eta^2 -\frac{9}{16}\left(\sum_{i=1}^n\frac{1}{h_i^3}\right)\eta^3 + \mathcal{O}(\eta^4),
\end{equation*}
when the driving noise follows either the NIG or GAL distributions.
\end{proof}

\subsection{Proof of Theorem 3.2.} \label{app:theo2}

\begin{proof}

We will leverage on the conditional representation of the non-Gaussian random vector $\mathbf{x}$ shown in eq. (5) of the main document and on the chain rule for KLDs \citep{informationtheory}:
\begin{flalign} \label{eq:KLD1}
 KLD(\pi(\mathbf{x,V}) \ || \ \pi^{Sym}(\mathbf{x,V})) = \\ 
&& \hspace{-2cm}  KLD(\pi(\mathbf{V}) \ || \ \pi^{Sym}(\mathbf{V})) + E_{\pi(\mathbf{V})}\left[KLD(\pi(\mathbf{x|V}) \ || \ \pi^{Sym}(\mathbf{x|V}))\right] \nonumber
\end{flalign}
Both terms on the right-hand side (RHS) of the previous expression can be evaluated. Also an upper bound can be found for the KLD we want to compute:
\begin{equation}\label{eq:KLD2}
KLD(\pi(\mathbf{x})||\pi^{Sym}(\mathbf{x})) \leq KLD(\pi(\mathbf{x,V})||\pi^{Sym}(\mathbf{x,V})).
\end{equation}
This inequality can be obtained by switching the vectors $\mathbf{x}$ and $\mathbf{V}$ in eq. \eqref{eq:KLD1} to get:
\begin{flalign} \label{eq:KLD3}
 KLD(\pi(\mathbf{x}) \ || \ \pi^{Sym}(\mathbf{x}))= \\ 
&& \hspace{-2cm}  KLD(\pi(\mathbf{x,V})\ || \ \pi^{Sym}(\mathbf{x,V})) - E_{\pi(\mathbf{x})}\left[KLD(\pi(\mathbf{V|x})\ || \ \pi^{Sym}(\mathbf{V|x}))\right] \nonumber
\end{flalign}
As $KLD(\pi(\mathbf{V|x}) \ || \ \pi^{Sym}(\mathbf{V|x})) \geq 0$, the inequality in eq. \eqref{eq:KLD2} follows. Next, considering $\tilde{\sigma} = 1/\sqrt{1+\eta \zeta^2}$ we restate the flexible and base models:
\begin{align*}
\text{Flexible model:} \ \  &\pi( \mathbf{x} | \boldsymbol{V} ) \sim \text{N}\left(\boldsymbol{m} = \tilde{\sigma}\zeta \mathbf{D}^{-1}(\boldsymbol{V}-\boldsymbol{h}), \ \mathbf{\Sigma} = \tilde{\sigma} \mathbf{D}^{-1} \text{diag}(\boldsymbol{V})^{-1} \mathbf{D}^{-T}\right)  \\ \\
\text{Base model:}  \ \ &\pi^{Sym}(\mathbf{x}|\boldsymbol{V}) \sim \text{N}\left(\boldsymbol{m}_S = \boldsymbol{0}, \ \mathbf{\Sigma}_S =  \mathbf{D}^{-1} \text{diag}(\boldsymbol{V})^{-1} \mathbf{D}^{-T}\right) \\ \\
\text{In both cases:} \ \ &\pi(V_{i}) \overset{d}{=} \pi^{Sym}(V_{i})  \overset{ind.}{\sim}  \begin{cases}
      \text{IG}(h_i,\eta^{-1} h_i^2)   &\text{(NIG Noise)} \\
      \text{Gamma}(h_i \eta^{-1},\eta^{-1})  &\text{(GAL Noise).}
    \end{cases}   
\end{align*}

Since $\pi(\mathbf{V}) \overset{d}{=} \pi^{Sym}(\mathbf{V})$, then $KLD(\pi(\mathbf{V})||\pi^{Sym}(\mathbf{V})) = 0$, and so to compute \linebreak $KLD(\pi(\mathbf{x,V})||\pi^{Sym}(\mathbf{x,V}))$ we only need to find the second term on the RHS of \mbox{eq. \eqref{eq:KLD1}} which requires the KLD between two Gaussian distributions:
\begin{flalign}\label{eq:KLDsimple}
 &KLD(\pi(\mathbf{x|V})||\pi^{Sym}(\mathbf{x|V}))  \nonumber \\*
 &= {\frac {1}{2}}\left(\operatorname {tr} \left(\mathbf{\Sigma}_S^{-1} \mathbf{\Sigma}\right)+\left(\boldsymbol{m}_S -\boldsymbol{m}\right)^{\mathsf {T}}\mathbf{\Sigma}_S^{-1}\left(\boldsymbol{m}_S-\boldsymbol{m}\right)- n +\log \left({\frac { |\mathbf{\Sigma}_S| }{|\mathbf{\Sigma}|}}\right)\right) \nonumber \\*
&= \frac{1}{2}\left( \frac{n}{1+\eta\zeta^2} + \frac{\zeta^2}{1+\eta\zeta^2} \left(\sum_{i=1}^nV_i-2\sum_{i=1}^nh_i+\sum_{i=1}^n \frac{h_i^2}{V_i}\right) -n+ n\log\left(1+\eta\zeta^2\right)\right)
\end{flalign}

According to eq. \eqref{eq:KLD1}, we now take the expectation of the previous expression w.r.t. $\pi(\mathbf{V})$. When the driving noise follows the NIG distribution, the mixing variables satisfy $V_i \sim \text{IG}(h_i, \eta^{-1}h_i^2)$, the expectations are $E[V_i]=h_i$, and $E[1/V_i]= (\eta + h_i)/h_i^2$ and the previous KLD simplifies to
$$KLD(\pi(\mathbf{x,V})||\pi^{Sym}(\mathbf{x,V})) =  E_{\pi(\mathbf{V})}\left[KLD(\pi(\mathbf{x|V})||\pi^{Sym}(\mathbf{x|V}))\right] = \frac{n}{2}\log \left(1+\eta\zeta^2\right).$$
The result of the Theorem follows from replacing the previous expression in  eq. \eqref{eq:KLD2} and then applying the inequality $\log \left(1+\eta\zeta^2\right) \leq \eta\zeta^2$. For GAL noise, the mixing variables follow the gamma distribution, thus $E[V_i]=h_i$ and  
\begin{equation*}
E\left[\frac{1}{V_i}\right] =  \begin{cases}
      \frac{1}{h_i-\eta}  &\text{if} \ \ \eta < h_i \\
      \text{Indeterminate}  &\text{if} \ \ \eta \geq h_i.
    \end{cases}
\end{equation*}
Therefore $KLD(\pi(\mathbf{x,V})||\pi^{Sym}(\mathbf{x,V}))$ only exists when $\eta < \min_{i=1,\dotsc,n} h_i$. If this condition is met, it can be shown that
$$KLD(\pi(\mathbf{x,V})||\pi^{Sym}(\mathbf{x,V})) \leq \frac{n}{2}\log \left(1+\eta\zeta^2\right),$$ by taking the expectation of eq. \eqref{eq:KLDsimple} w.r.t. $\pi(\mathbf{V})$ and hence the result also follows in this case.

\end{proof}

\section{Penalizing the probability of large marginal events in Matérn and OU processes} \label{sect:scalingetaMatern}
The marginals of the process $X(t)$ defined in Appendix \ref{app:ngma}, will be non-Gaussian and exhibit heavier tails when $\eta>0$. Thus, we can calibrate the prior by penalizing the probability of large events in the marginals: $P(|X(t)|>3\sigma_{marg}(\sigma,\kappa))$, where $\sigma_{marg}$ is the marginal standard deviation. For Matérn models, it is $$\sigma_{marg}(\sigma,\kappa) = \sigma \kappa^{-(\alpha-d/2)}\sqrt{\Gamma(\alpha-d/2)/(\Gamma(\alpha)(4\pi)^{d/2})}.$$
The marginal excess kurtosis for the Matérn model depends on $\kappa$, the spatial range parameter, for instance, for $\alpha=2$ in 1D, the excess kurtosis is proportional to $\kappa\eta$, in the symmetric case. Therefore, we should take the range parameter $\kappa$ into account when calibrating the prior. We compute how many more large marginal events we are expected to have compared with the Gaussian case:
\begin{equation} \label{eq:ratioQ}
Q(\eta,\kappa) =  \frac{P(|X_{\eta,\kappa}(t)|>3\sigma_{marg}(1,\kappa))}{P(|X_{\eta=0,\kappa}(t)|>3\sigma_{marg}(1,\kappa))} , 
\end{equation}
where $X_{\eta,\kappa}(t)$ is the marginal distribution of the continuous process $X$ with spatial range parameter $\kappa$ driven by symmetric non-Gaussian noise with parameters $\eta$ and $\sigma=1$, and $X_{\eta=0,\kappa}$ is its Gaussian counterpart. We can compute the probabilities in eq. \eqref{eq:ratioQ} based on the characteristic functions of the marginals $X(t)$ given in Appendix C, and then the calibration is done by setting a low probability $\alpha_\eta$ on the event $Q(\eta,\kappa) > 2$: $$P(Q(\eta,\kappa) > 2)= \alpha_\eta \ \ \longrightarrow \ \ P(\eta > Q^{-1}(2|\kappa)) = \alpha_\eta,$$
where $Q^{-1}(2|\kappa)$ is defined as the root of the equation $Q(\eta,\kappa)-2$ w.r.t. $\eta$ while keeping $\kappa$ fixed. The rate parameter of the exponential prior distribution for $\eta$ will then be $\theta_\eta(\kappa) = -\log(\alpha_\eta)/(Q^{-1}(2|\kappa))$.
Numerical computations showed that $Q^{-1}(2|\kappa)$ follows with high accuracy the functions of $\kappa$ in Table \ref{tab:qm1}. The calibration then boils down to choosing the probability $\alpha_\eta$, which is the probability that the marginals will have twice as many large marginal events compared with the Gaussian case. This user-defined probability governs the contraction towards the Gaussian model (the closer $\alpha_\eta$ is to 0, the higher is the contraction). 

\begin{table}
\center
\begin{tabular}{ccc}
\hline
\textbf{Model}     & \textbf{NIG}         & \textbf{GAL}\\ \hline 
OU $d=1$               & $0.1566 \kappa^{-1}$ & $0.1540 \kappa^{-1}$ \\
Matérn $d=1$, $\alpha=2$ & $0.2676\kappa^{-1}$  & $0.2488\kappa^{-1}$  \\
Matérn $d=2$, $\alpha=2$ & $0.2513\kappa^{-2}$   & $0.2513\kappa^{-2}$  \\ \hline
\end{tabular}
\caption{$Q^{-1}(2|\kappa)$ for several stationary models and driving noises, where $d$ is the spatial dimension and $\alpha_\eta$ is the smoothness parameter.}
\label{tab:qm1}
\end{table}

\section{Additional simulation figures and studies}\label{app:simfigs}


Fig. \ref{fig:sim2} shows generated sample paths for the different simulation scenarios. The following Figs. \ref{fig:sim3} and \ref{fig:sim32} contain box plots, each built from the 200 posterior means and width of credible intervals (based on the 0.5 and 0.95 posterior quantiles) obtained from the 200 simulation replicates of the first simulation set. More details are found in section 5 of the main paper.

\begin{figure}[]
    \centering
 \includegraphics[width=0.45\linewidth]{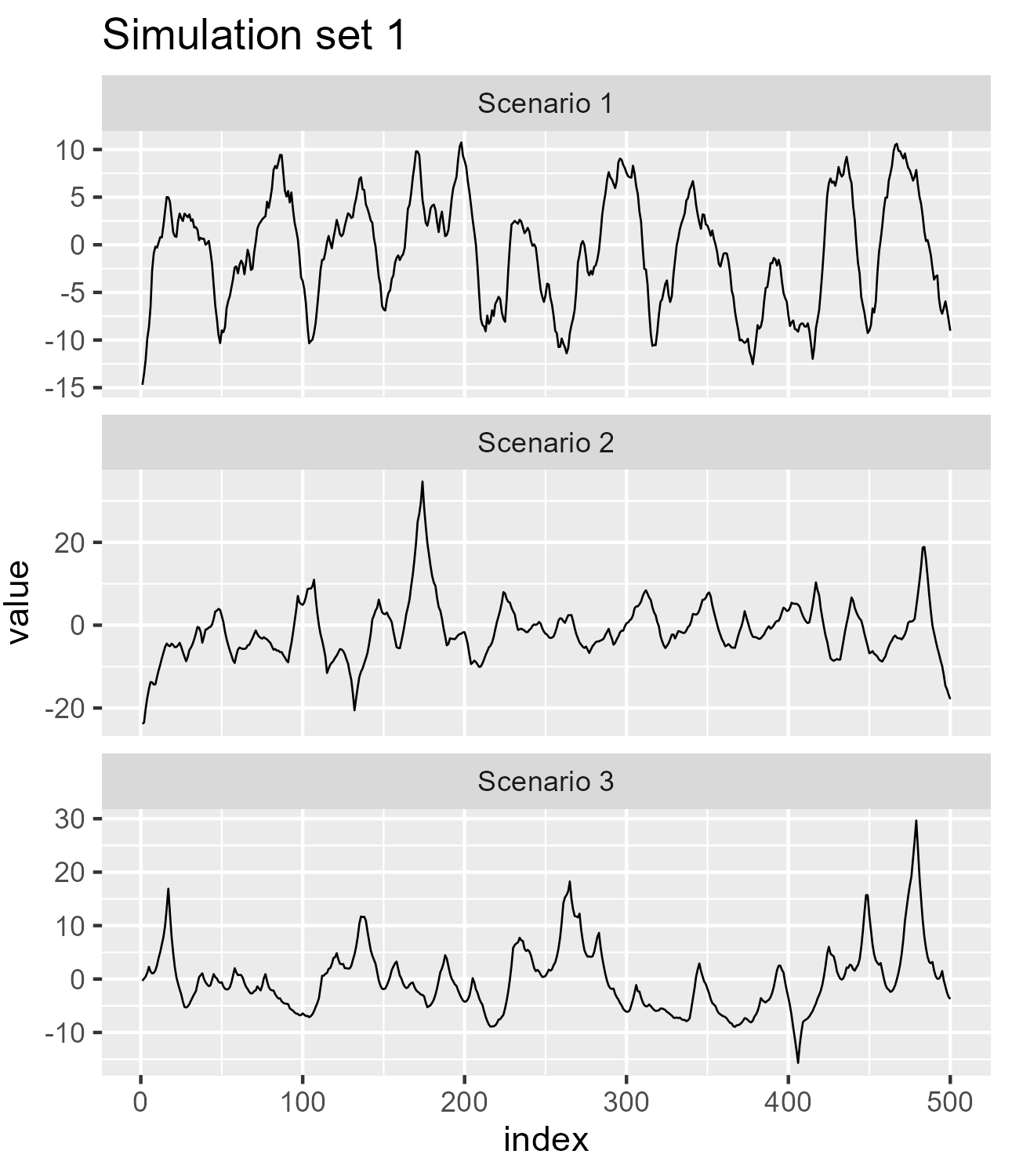}
  \includegraphics[width=0.45\linewidth]{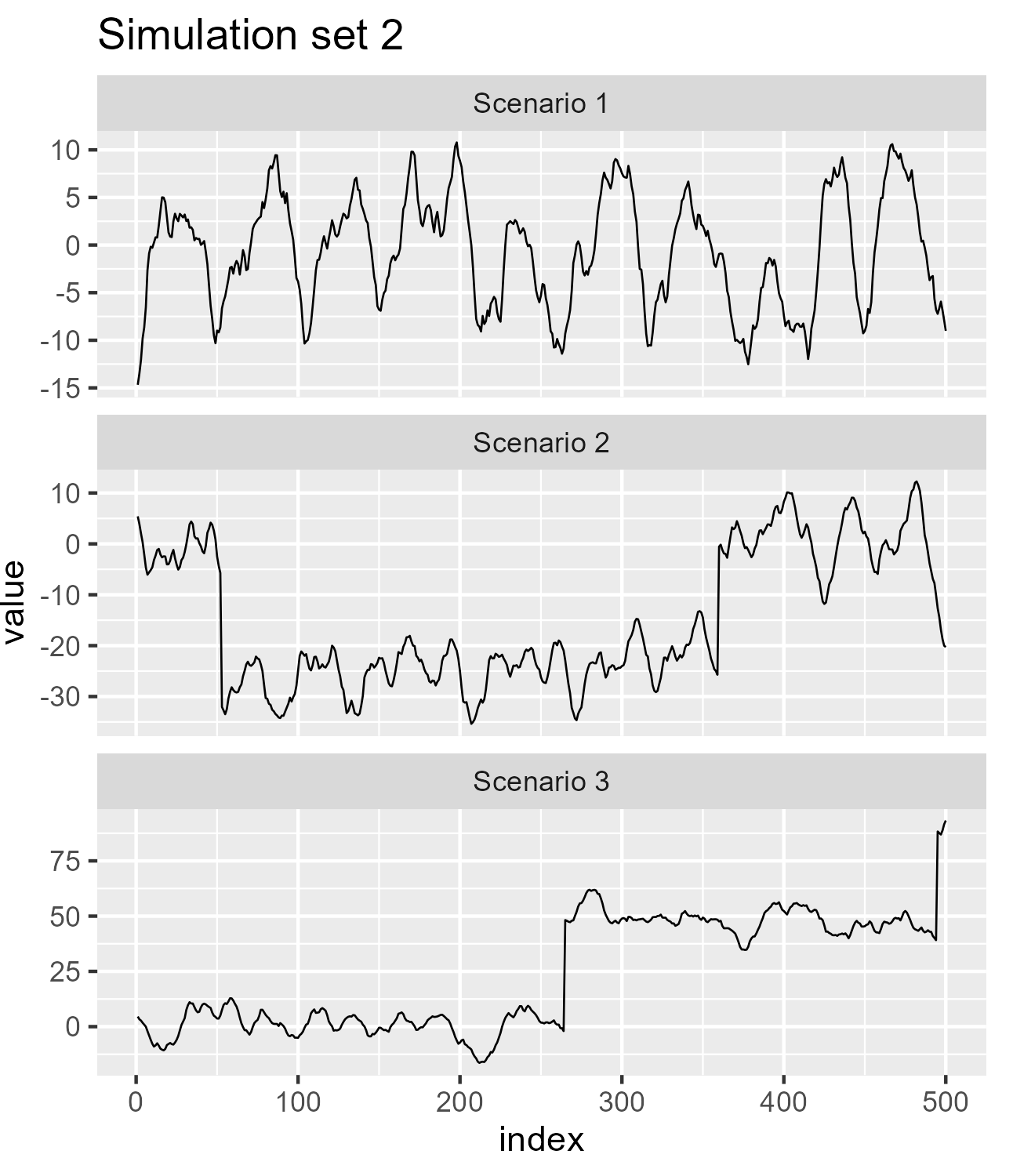}
  \caption{Sample paths for the 2 simulation sets and scenarios.}
  \label{fig:sim2}
\end{figure}

\begin{figure}[]
    \centering
 \includegraphics[width=0.8\linewidth]{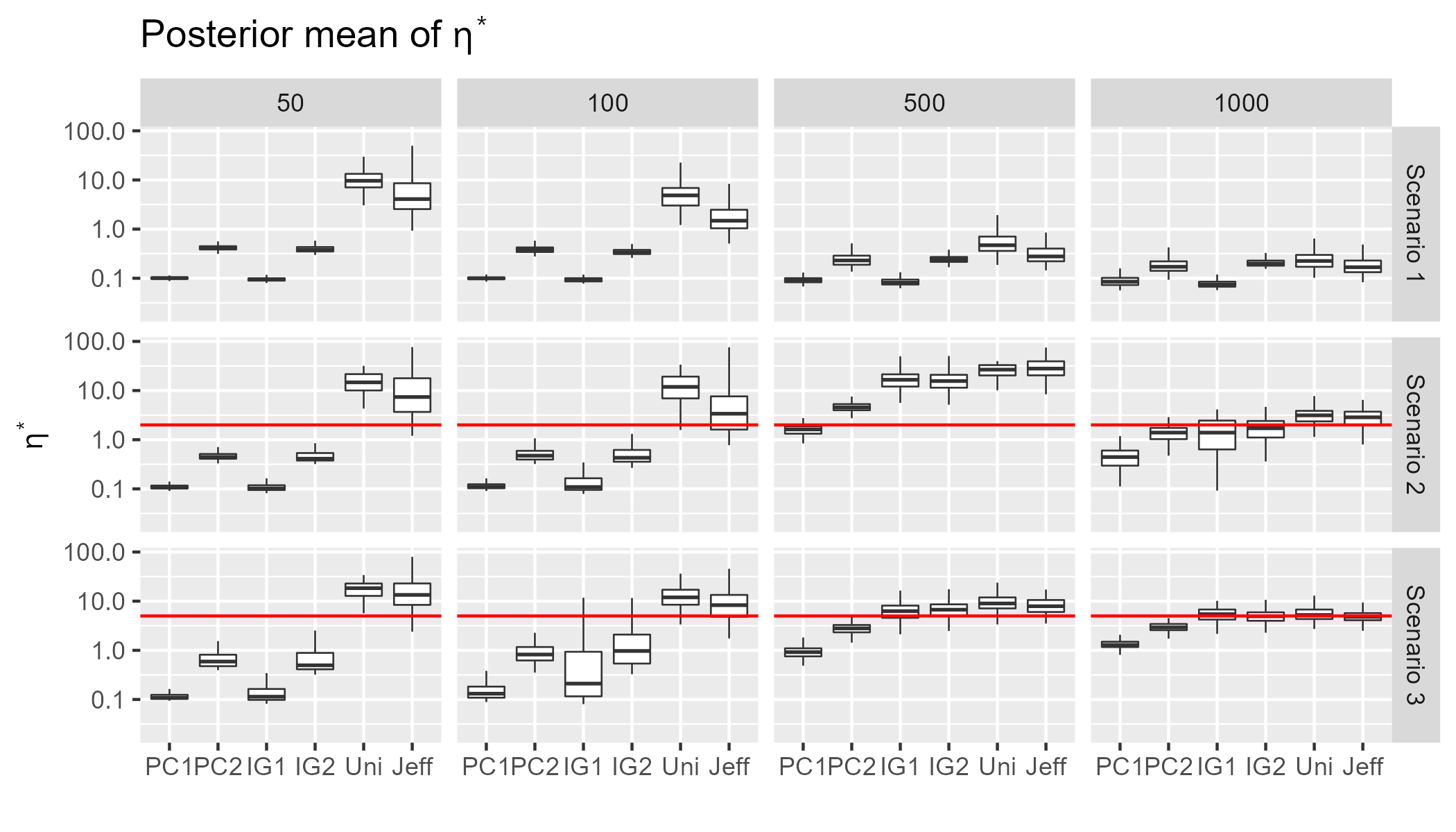}
   \includegraphics[width=0.8\linewidth]{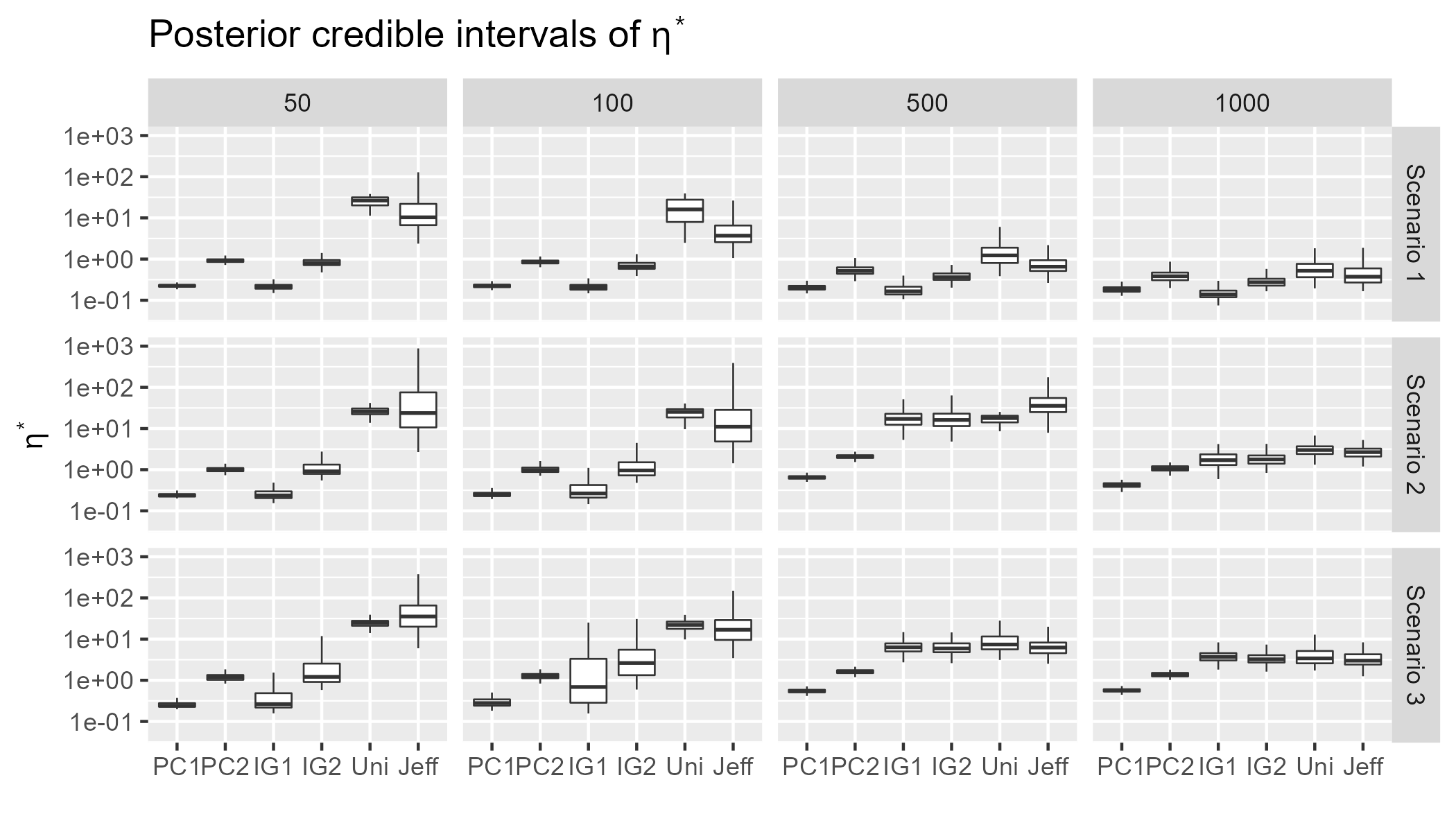}
  \caption{Histograms of the posterior means (top) and width of the credible intervals (bottom) of $\eta^\star$ and for different sample sizes, prior configurations, and scenarios of simulation set 1.}
  \label{fig:sim3}
\end{figure}

\begin{figure}[]
    \centering
     \includegraphics[width=0.8\linewidth]{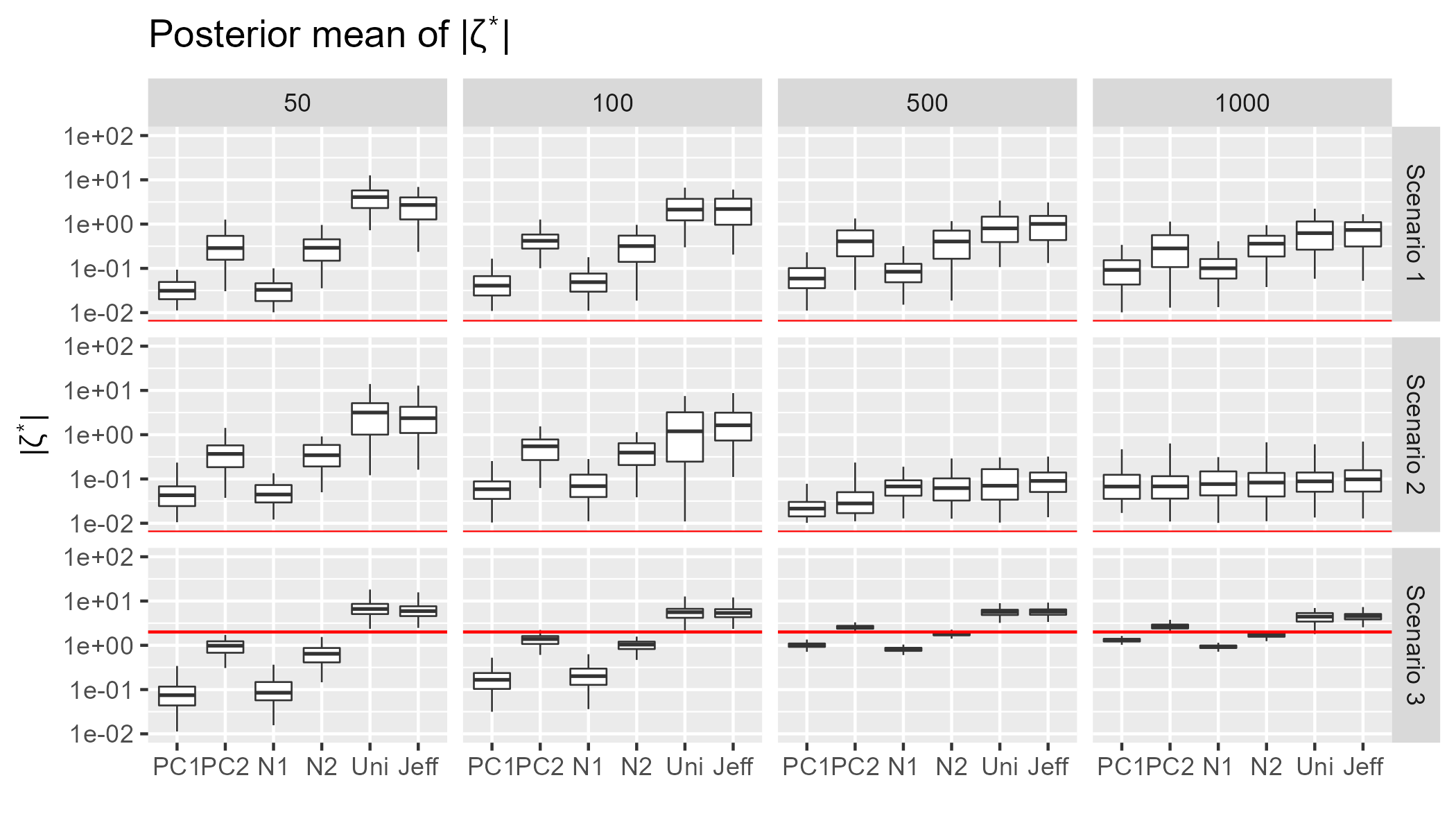}
  \includegraphics[width=0.8\linewidth]{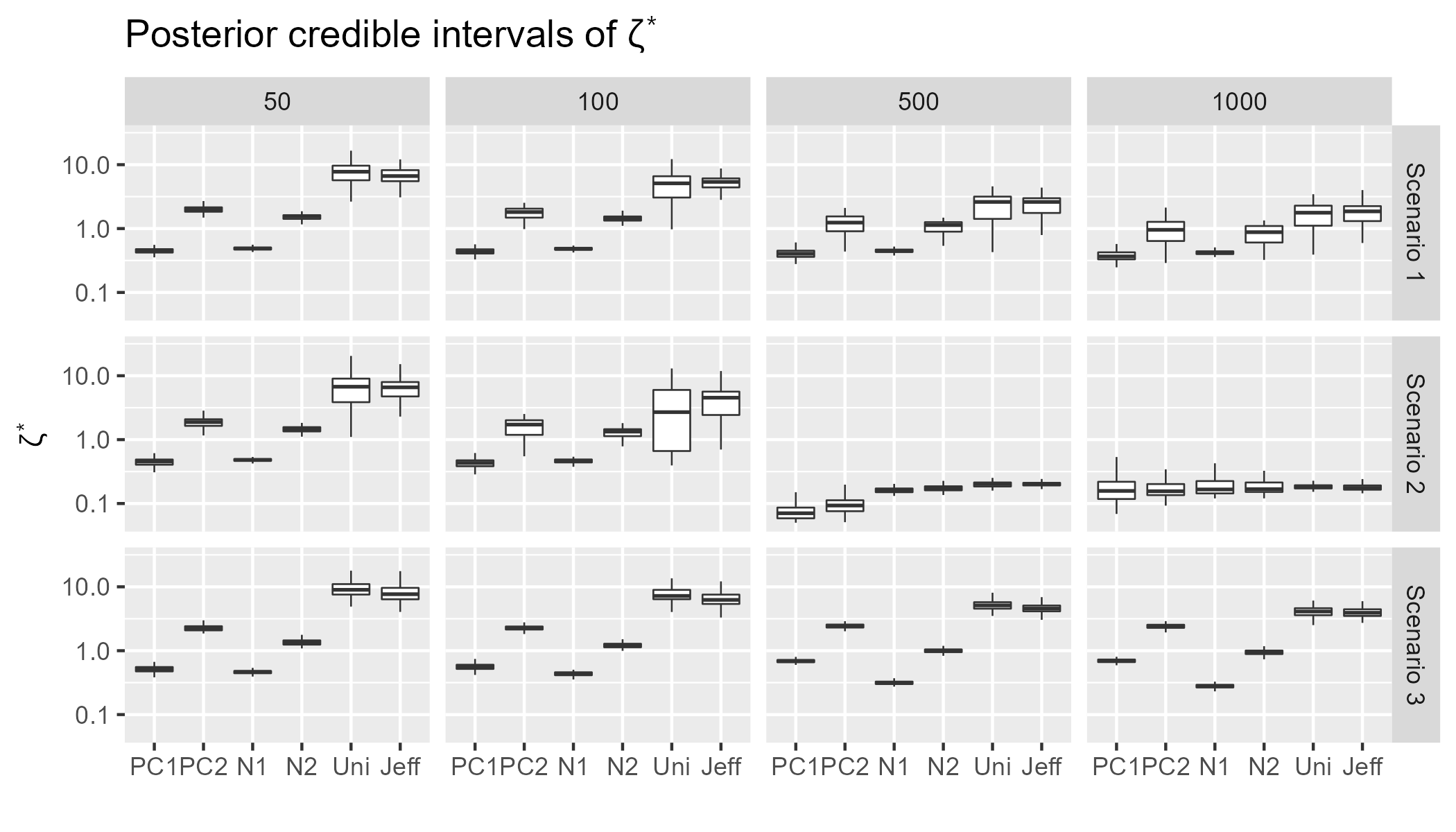}
  \caption{Histograms of the posterior means (top) and width of the credible intervals (bottom) of $\zeta^\star$ and for different sample sizes, prior configurations, and scenarios of simulation set 1.}
  \label{fig:sim32}
\end{figure}

\subsection{Results for simulation set 2}
          
The posterior means of $\eta^\star$ (Fig. \ref{fig:simsce2}) show a large sensitivity of the model to the addition of jumps in the latent field when using the inverse gamma, Jeffreys, or uniform priors. Namely, the fitted models are closer to ones with Cauchy driving noise, while this sensitivity is reduced when using PC priors. Note that the NIG distribution converges to the Cauchy distribution when $\eta \to \infty$. Also, when using the PC priors, the widths of the posterior credible intervals of $\eta^\star$ are considerably smaller, and the posterior means of $\zeta^\star$ are closer to 0.

\begin{figure}[]
    \centering
 \includegraphics[width=\linewidth]{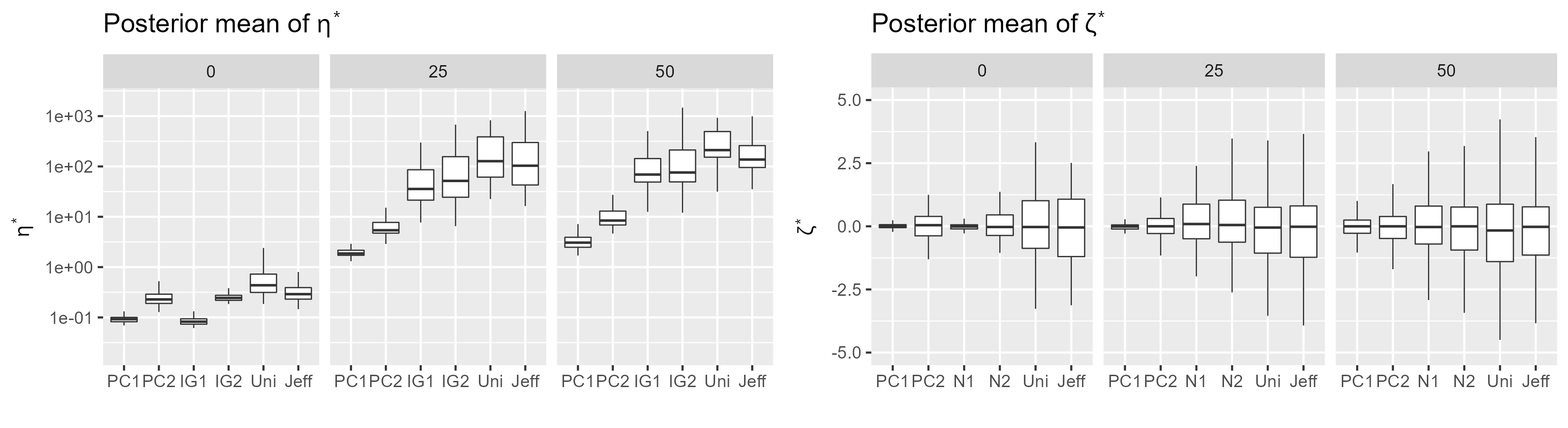}
  \includegraphics[width=\linewidth]{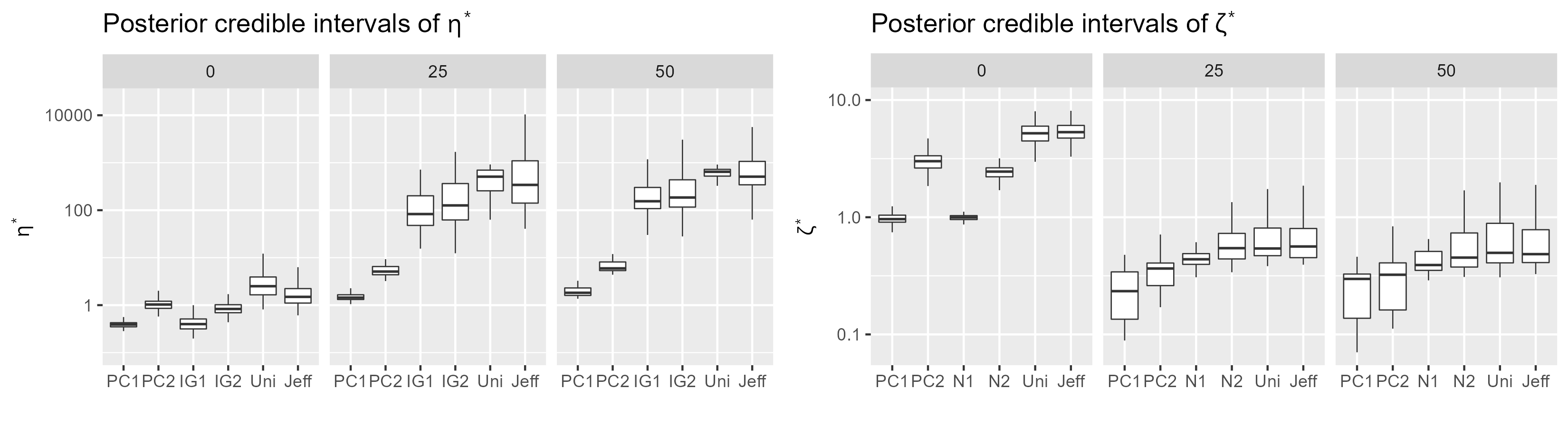}
  \caption{Histograms of the posterior means (top) and widths of the posterior credible intervals (bottom) for $\eta^\star$ (left) and $\zeta^\star$ (right) and for different prior configurations and scenarios of the simulation set 2. The size of the added jumps is shown at the top of each subfigure.}
  \label{fig:simsce2}
\end{figure}




\end{appendices}

\end{document}